\documentclass[12pt]{article}
\usepackage{eurosym}
\usepackage{amsmath,amsfonts,amssymb,amsthm}
\usepackage{lscape}
\usepackage{graphicx}
\usepackage{natbib}
\usepackage{amsmath,amssymb,graphicx,setspace,rotating}
\usepackage[outdir=./]{epsfig} 
\usepackage{pdflscape}
\usepackage[section]{placeins}   
\usepackage{color}  
\usepackage{threeparttable} 
\usepackage{bm} 
\usepackage[linesnumbered]{algorithm2e}
\usepackage[multiple]{footmisc}
\usepackage{caption}
\DeclareMathOperator*{\plim}{plim}

\newtheorem{definition}{Definition}
\newtheorem{proposition}{Proposition} 
\newtheorem{assumption}{Assumption}  
\newtheorem{lemma}{Lemma}

\makeatletter
\newcommand*\bigcdot{\mathpalette\bigcdot@{.5}}
\newcommand*\bigcdot@[2]{\mathbin{\vcenter{\hbox{\scalebox{#2}{$\m@th#1\bullet$}}}}}
\makeatother

\newcommand{\vect}[1]{\boldsymbol{#1}}

\def\1{1\!{\rm l}}

\def\y{\mathbf{y}}
\def \zb{\mathbf{z}}

\newcommand{\dt}{\text{d}}

\newcommand{\blind}{1}

\begin{document}
\def\spacingset#1{\renewcommand{\baselinestretch}%
	{#1}\small\normalsize} \spacingset{1}

\date{\today}

\if1\blind
{\title{\bf Robust Approximate Bayesian Inference  with Synthetic Likelihood}
	
  \author{David T. Frazier\thanks{
	Department of Econometrics and Business Statistics, Monash University, and the Australian Centre of Excellence for Mathematical and Statistical Frontiers}
	 \and
	Christopher Drovandi\thanks{School of Mathematical Sciences, Queensland University of Technology and the Australian Centre of Excellence for Mathematical and Statistical Frontiers}
}
  \maketitle
} \fi

\if0\blind
{
	\bigskip
	\bigskip
	\bigskip
	\begin{center}
		{\LARGE\bf Robust Approximate Bayesian Inference  with Synthetic Likelihood}
	\end{center}
	\medskip
} \fi

\bigskip

\begin{abstract}
Bayesian synthetic likelihood (BSL) is now an established method for conducting approximate Bayesian inference in models where, due to the intractability of the likelihood function, exact Bayesian approaches are either infeasible or computationally too demanding. Implicit in the application of BSL is the  assumption that the data generating process (DGP) can produce simulated summary statistics that capture the behaviour of the observed summary statistics. We demonstrate that if this compatibility between the actual and assumed DGP is not satisfied, i.e., if the model is misspecified, BSL  can yield unreliable parameter inference. To circumvent this issue, we propose a new BSL approach that can detect the presence of model misspecification, and simultaneously deliver useful inferences even under significant model misspecification. Two simulated and two real data examples demonstrate the performance of this new approach to BSL, and document its superior accuracy over standard BSL when the assumed model is misspecified. 
\end{abstract}


\spacingset{1.5} 

\section{Introduction}

In situations where the likelihood of the underlying model is intractable, approximate Bayesian methods are often the only feasible solution to conduct Bayesian inference. Indeed, approximate Bayesian methods are an increasingly common tool in the arsenal of the practicing statistician and  allow users to conduct reliable inference in models where exact Bayesian inference procedures are either infeasible, or too computationally demanding.

The literature on approximate Bayesian inference now includes several competing approximate methods that are often useful in different scenarios. Arguably, the two most common likelihood-free Bayesian methods in the {statistical} literature are approximate Bayesian computation (ABC) (see, e.g., \citealp{marin2012} for a review) and Bayesian synthetic likelihood (BSL) (\citealp{wood2010}, \citealp{price2018bayesian}).  {The machine learning community is also making significant contributions to likelihood-free methods, such as using emulation to reduce the number of calls to the model simulator (e.g.\ \citealp{Gutmann2016}) and training neural conditional density estimators such as normalizing flows (e.g.\ \citealp{Papamakarios2018}). We refer to \citet{Cranmer2019}  for a comprehensive review of machine learning approaches to likelihood-free methods.}  

Following the frequentest synthetic likelihood approach of \cite{wood2010}, \cite{price2018bayesian} develop an alternative to ABC by constructing a Bayesian version of synthetic likelihood, which places a prior distribution over the parameters and generates an approximate posterior. Unlike ABC, which implicitly estimates a version of the likelihood for the summaries, BSL directly assumes that the joint density of the summary statistics, conditional on the unknown model parameters, is Gaussian with unknown mean and variance. Using independent simulations obtained from the assumed data generating process (DGP), the mean and variance of the summary statistics are then estimated, and used to construct a (simulated) Gaussian likelihood function that is directly inserted into standard Markov Chain Monte Carlo (MCMC) algorithms. \cite{price2018bayesian} demonstrate the BSL approach across several examples, and show that it often performs well in comparison with ABC.

BSL, and approximate Bayes methods more generally, are most often applied in situations where the complexity of the model that is assumed to have generated the observed data renders exact Bayesian inference infeasible. That is, by the very nature of the problems to which BSL is commonly applied, the model is so complex that we can not easily access the DGP and must instead resort to an approximate inference approach. However, while complicated, highly-structured models that allow for vast complexity allow us to explain critical features of the observed data, it is unlikely that any modeler will be able to construct an entirely accurate model that captures all features of the observed data. In short, all models are wrong and the scientist cannot obtain a ``correct'' one through excessive elaboration (\citealp{box1976science}).

The implications of such a statement are particularly worrying in the context of BSL, where the assumed underlying DGP is often very complex. Indeed, applying the above reasoning of Box, it must be the case that the models to which BSL is routinely applied are misspecified representations of the actual, or ``true'', DGP. In such situations, the application of BSL deserves further scrutiny given the recent results of \cite{frazier2020model}, which demonstrate that if the model is misspecified the ABC posterior can be ill-behaved. Given that the principles underlying ABC and BSL are qualitatively the same, further analysis is needed to ensure that BSL does not suffer from the same issues as ABC in cases where the model is misspecified.

Through several simulated and empirical examples, we demonstrate that if the assumed model is misspecified, point estimators and credible sets obtained from BSL are unreliable. {To circumvent this issue, we propose two novel versions of BSL that deliver ``robust'' inferences regardless of whether the model is correctly specified. Herein, we follow the robust statistics literature, as described in, e.g., \cite{hampel2011robust}, and consider a statistical inference procedure to be robust if it is not overly ``sensitive'' to departures from the underlying modeling assumptions.\footnote{Formalizing this notion of robustness requires a refined mathematical treatment that is beyond the scope of this paper. A precise definition of robustness requires the specification and use of infinitesimal neighborhoods that capture the degree of model misspecification, and which then allow us to formally define a notion ``sensitivity'' (see, e.g., \citealp{hampel2011robust}, and \citealp{rieder2012robust} for a discussion). Given this, we leave a formal study on the theoretical robustness of this new BSL approach for future research.}}

{This new BSL approach has three principle benefits over standard BSL. Firstly, this new procedure is less sensitive to model misspecification than standard BSL. In particular, the resulting posteriors are less affected by model misspecification than those obtained from BSL (see Section \ref{sec:diffs} for a detailed discussion and Section \ref{sec:examples} for specific examples). Consequently, this approach yields more reliable point estimators and uncertainty quantification in misspecified models. Second, this new BSL approach has an in-built mechanism for diagnosing model misspecification, which allows us to discern which components of the model may in-fact be misspecified. Lastly, this new approach is computationally robust in the following sense: when the model is misspecified, the standard BSL posterior can require an excessive number of model simulations to generate accurate samples, however, our proposed BSL approach does not suffer from this issue. Given the above notions of robustness, both statistical and computational, we refer to this new approach as robust BSL (R-BSL).}

{Through a series of examples, both simulated and empirical, we demonstrate that R-BSL yields reliable statistical inferences regardless of whether the model is correctly or incorrectly specified, and can consistently detect when the modeling assumptions are violated. Repeated sampling results demonstrate that R-BSL yields more accurate point estimators and quantifies uncertainty better than BSL in misspecified models. In addition, theoretical (and simulation results) demonstrate that if the model is correctly specified, R-BSL and BSL deliver similar statistical inferences. As such, R-BSL allows users to hedge against model misspecification, but ensures that the resulting inference remains accurate if the model is correctly specified.}

{Before moving on, we remark that throughout this paper, we assume that the Gaussian assumption of the synthetic likelihood is at least approximately correct. We refer the reader to the discussion in Section \ref{sec:discussion} for more details, and potential impacts on the resulting inferences when this assumption is invalid.}

The remainder of the paper is organized as follows. In Section two we give a brief overview of BSL and examine the consequences of model misspecification in BSL. Section three presents our robust approach to BSL. Through a sequence of examples discussed in Section four, we document the good performance of R-BSL and the poor performance of BSL across models with varying levels of misspecification. In addition, Section four contains an empirical application to the analysis of invasive species. {Section five concludes. Additional examples and the proofs of the technical results are given in the supplementary material.}

\section{Bayesian Synthetic Likelihood and Compatibility}
\subsection{Bayesian Synthetic Likelihood Framework}
We observe data ${\y}=(y_{1},\dots,y_{n})^{\top}$, $n\geq 1$, and denote by $P^n_0$ the true distribution of the observed sample. The true distribution is unknown and instead we consider that the class of probability measures  
$\{\theta\in\Theta\subset \mathbb{R}^{d_{\theta }},\;n\geq1:P^{n}_{{\theta }}\}$, for some value of $\theta$, have generated the data, and denote the corresponding conditional density as $p_n(\cdot|\theta)$. Given prior beliefs over the unknown parameters in the model $\theta$, represented by the probability measure $\Pi(\theta)$, with its density denoted by $\pi({\theta })$, our aim is to
produce draws from the exact posterior density 
$$
\pi({\theta \mid \y})\propto p_n(\y|\theta)\pi({\theta }).
$$ 

In situations where the likelihood is intractable, sampling from $\pi({\theta \mid \y})$ can be computationally costly or infeasible, however, so-called likelihood-free methods can still be used to conduct inference on the unknown parameters $\theta$. The most common implementations of these methods in the statistical literature are approximate Bayesian computation (ABC) and Bayesian synthetic likelihood (BSL). Both ABC and BSL generally degrade the data down to a vector of summary statistics and then perform posterior inference on the unknown $\theta$, conditional only on this vector of summary statistics.

More formally, let $\eta(\cdot):\mathbb{R}^{n}\rightarrow\mathbb{R}^{d_\eta}$ denote a $d_\eta$-dimensional map, $d_\eta\ge d_\theta$, that represents the chosen summary statistics, and let $\zb:=(z_1,\dots,z_n)^{\top}\sim P_\theta^n$ denote data simulated from the model $P_\theta^n$. For $G_n(\cdot|\theta)$ denoting the projection of $P^n_\theta$ under $\eta(\cdot):\mathbb{R}^n\rightarrow\mathbb{R}^{d_\eta}$, with $g_n(\cdot|\theta)$ its corresponding density, the goal of approximate Bayesian methods is to generate samples from the approximate or `partial' posterior $$\pi[{\theta \mid \eta(\y)}]\propto g_n[{\eta(\y)\mid\theta }]\pi(\theta).$$ 
However, given the complexity of the assumed model, $P_\theta^n$, it is unlikely that the structure of $G_n(\cdot|\theta)$ is any more tractable than the original likelihood function $p_n(\y|\theta)$. Therefore,  simulation-based sampling schemes must be applied to generate samples from $\pi[\theta\mid{\eta(\y)}]$.

The approximate methods of ABC and BSL differ in how $g_n(\cdot|\theta)$ is estimated. ABC forms an implicit nonparametric estimator of $g_n(\cdot|\theta)$, while BSL uses a parametric or semi-parametric \citep{an2018robust} approximation of $g_n(\cdot|\theta)$.{In particular, BSL replaces - an estimate of - the (intractable) density $g_n(\cdot|\theta)$ by a multivariate Gaussian approximation: $$\mathcal{N}\left[\eta;\mu(\theta),\Sigma(\theta)\right],$$where $\mu(\theta)$ and $\Sigma(\theta)$ denote the mean and variance of the summary statistics. In cases where $\mu(\theta),\Sigma(\theta)$ are known we can obtain the ``exact'' BSL posterior}
\begin{flalign*}
{\pi}[\theta\mid\eta(\y)]&\propto\mathcal{N}\left[\eta(\y);\mu(\theta),\Sigma(\theta)\right]\pi(\theta).\nonumber
\end{flalign*}{However, in almost any practical example $\mu(\theta)$ and $\Sigma(\theta)$ are unknown and we must replace these quantities with the estimated counterparts $\mu_m(\theta)$ and $\Sigma_m(\theta)$, obtained as} 
\begin{flalign*}
\mu_m(\theta)&=\frac{1}{m}\sum_{i=1}^{m}\eta(\zb^i),\;\;\Sigma_m(\theta)=\frac{1}{m}\sum_{i=1}^{m}\left[\eta(\zb^i)-\mu_m(\theta)\right]\left[\eta(\zb^i)-\mu_m(\theta)\right]^{\top},
\end{flalign*}and where each simulated data set $\zb^i$, $i=1,\dots,m$, are generated iid from $P^n_{\theta}$; that is, both $\mu_m(\theta)$ and $\Sigma_m(\theta)$ depend on the simulated data sets $\zb^1,\dots,\zb^m$. The Gaussian approximation is then directly used within an MCMC sampling scheme to sample from the following approximation to the partial posterior, hereafter referred to as the BSL posterior, 
\begin{flalign}\label{eq:BSL_post}
\hat{\pi}[\theta\mid\eta(\y)]&\propto\bar{g}_n[{\eta(\y)\mid\theta }]\pi(\theta),
\\
\bar{g}_n[{\eta(\y)\mid\theta }]&:=\int \mathcal{N}\left[\eta(\y);\mu_m(\theta),\Sigma_m(\theta)\right]\left\{\prod_{i=1}^{m}g_n[\eta(\zb^i)\mid\theta]\right\}\text{d}\eta(\zb^{1})\cdots\text{d}\eta(\zb^m).\nonumber
\end{flalign}

{\citet{price2018bayesian} demonstrate empirically that the BSL posterior depends weakly on $m$, provided that $m$ is chosen large enough so that the plug-in synthetic likelihood estimator has a small enough variance to ensure that MCMC mixing is not adversely affected.  In this paper we choose $m$ and the number of MCMC iterations large enough so that the Monte Carlo error arising from MCMC is small.}

Due to the parametric nature of equation \eqref{eq:BSL_post}, BSL can often treat summary statistics of larger dimension than ABC and can lead to sharper inference in some cases. While the validity of the Gaussian approximation is often warranted if the underlying summaries satisfy a central limit theorem (\citealp{wood2010}), even in cases where the summary statistics are far from Gaussian, BSL has displayed some insensitivity to violations of this assumption (\citealp{price2018bayesian}). However, if the statistics are very far from being Gaussian, this can result in a significant loss of accuracy, see \citet{an2018robust} for a demonstration.  

\subsection{Model Incompatibility and its Consequences}

{BSL implicitly maintains that the assumed model can generate simulated summary statistics $\eta(\zb)$ that can match the observed summary statistics $\eta(\y)$.} That is, BSL is not required to match every aspect of the data, but only those features of the data that are captured via the summary statistics $\eta(\y)$. This differs from a standard Bayesian framework based on a likelihood, where, under general regularity conditions, the posterior ultimately gives higher probability mass to values of $\theta\in\Theta$ that ensure the Kullback-Leibler (KL) divergence
$$\mathcal { D } \left( P^n _ { 0 } \| P^n _ { \theta } \right) =\int \log \left\{ \frac { p^n _ { 0 } ( \mathbf { y } ) } { p_n(\y|\theta) } \right\} \dt P^n _ { 0 } ( \mathbf { y } ) $$ is as close to zero as possible. When the model is misspecified, i.e., when $P_0^n\not= P^n_\theta$ for any $\theta\in\Theta$, following \cite{kleijn2012bernstein}, the posterior eventually places increasing mass on the value $\theta^*\in\Theta$ that minimizes the KL-divergence. 

{Given that BSL attempts to simulate summary statistics $\eta(\zb)$ that can match the value of the observed summary statistic $\eta(\y)$, as measured by a weighted Euclidean norm, KL divergence is not the most meaningful notion of model misspecification associated with BSL. A more meaningful notion is whether or not $\eta(\zb)$ can match $\eta(\y)$ in terms of the Euclidean norm. Therefore, we follow \cite{marin2014relevant}, and \cite{frazier2020model}, and say that the assumed model is misspecified when it can not generate summaries $\eta(\zb)$ that can match $\eta(\y)$. More formally, for $b_0:=\plim_{n} \eta(\y)$ and $b(\theta):=\plim_{n} \eta(\mathbf{z})$, denoting the probability limits of the summaries as $n\rightarrow\infty$, we can state this notion of misspecification as follows.}\footnote{The fact that we require the summary statistics in BSL to concentrate to well-defined limit counterparts should not come as a surprise. \cite{frazier2016asymptotic} have demonstrated that this concentration is necessary to formally discuss the asymptotic behavior of ABC, while \cite{frazier2019bayesian} echos this finding in the specific context of BSL.} 
\begin{definition}\label{def1}
The model $P^n_\theta\times \Pi$ and summary statistic map $\eta(\cdot)$ are compatible  if $$\inf_{\theta\in\Theta}\|b(\theta)-b_0\|=0.$$
\end{definition}
\noindent {Heuristically, compatibility requires that asymptotically $\eta(\y)$ must be in the range of $\eta(\zb)$ and implies that, for some value of $\theta$, $\eta(\zb )$ can recover $\eta(\y)$ when $\zb $ is simulated under $P^n_\theta$. Compatibility is not concerned with the distributions of $\eta(\y)$ and $\eta(\zb)|\theta$, which ultimately must be degenerate if they are to be informative about $\theta$, but only their probability limits. Employing this notion of model misspecification allows us to analyze a large set of examples, since all this concept requires is that the summaries satisfy some weak law of large numbers.\footnote{It is also useful to point out that compatibility is not directly related to, or interpreted in terms of, any statistical divergence, and is only related to the minimum achievable distance between observed and simulated summaries (in the infinite data limit). While it may be possible to recast compatibility in terms of a statistical divergence, it is not clear what, if any, additional insights such an analysis would yield.}} 

When the model is not compatible, it can not (asymptotically) replicate the value of the observed statistics and, following the nomenclature in \cite{frazier2020model}, we say that the model is \textit{misspecified in the BSL sense}.\footnote{This notion of model misspecification is a ``global'' notion of misspecification, and is precisely the same notion of model misspecification used in \cite{marin2014relevant} and \cite{frazier2020model}. This is in contrast to the notion of ``local'' misspecification that is commonly entertained in the robust statistics literature.}
As recently discussed by \cite{frazier2020model}, in the context of ABC, when approximate methods are based on a model and summary statistic combination that are not compatible, the resulting posteriors can be ill-behaved and statistical inferences based on these posteriors can be highly-unreliable. Given that ABC and BSL are based on the same principles, it is highly likely that BSL will suffer from the same issues as ABC when the above compatibility condition is not satisfied. While we demonstrate this with realistic examples in Section \ref{sec:examples}, we first consider an artificially simple example where BSL should perform well, but due to model misspecification, BSL inference can be unreliable.

\subsection{Toy Example 1: Contaminated Normal Model}\label{sec:toynorm}

To demonstrate how BSL can fail under model misspecification, we consider an artificially simple example: our goal is inference on the unknown mean parameter $\theta$ in the assumed model 
\begin{equation}
y_i = \theta+{v}_i,\;{v}_i \stackrel{iid}{\sim}\mathcal{N}(0,1)\label{eq:mod1}.
\end{equation}
However, we  consider that the assumed model in \eqref{eq:mod1} is correct, but only for a portion of the data, $\omega\in(0,1)$, while the remaining portion of the data, $1-\omega$, is contaminated by outliers that are also normal but can have a much larger variance than unity. The true DGP can then be stated as the mixture model
\begin{equation}\label{eq:truemod}
y_i=\begin{cases}
\theta+\epsilon_{1,i},\;\epsilon_{1i}\sim\mathcal{N}(0,1),\;\text{ with probability } \omega\\\theta+\epsilon_{2,i},\;\epsilon_{1i}\sim\mathcal{N}(0,\sigma_\epsilon^2),\;\text{ with probability }1-\omega
\end{cases}	.
\end{equation}

Given that our assumed model is \eqref{eq:mod1}, the most reasonable set of summary statistics to choose are the sample mean $\eta_{1}(\mathbf{y})=\frac{1}{n}\sum_{i=1}^{n}{y}_{i}$ and sample variance $\eta_{2}(\mathbf{y})=\frac{1}{n-1}\sum_{i=1}^{n}({y}_{i}-\eta_{1}(\mathbf{y}))^{2}$. It can easily be shown that when $\sigma^2_\epsilon$ in the true DGP, \eqref{eq:truemod}, satisfies $\sigma^2_\epsilon\neq1$, the model is not compatible with the variance summary, and Definition \ref{def1} is not satisfied. 

Regardless of this model misspecification, one would think that BSL should perform well: the first summary statistic, $\eta_1(\y)$, is Gaussian, the statistic $\eta_2(\y)$ satisfies the central limit theorem, and $\eta_1(\y)$ is sufficient for $\theta$.  Moreover, the model is only wrong for a proportion of the data. 

For this experiment, we fix the level of data contamination at $1-\omega=0.20$, and generate simulated data
sets for $\y$ where each data set corresponds to a different value of
${\sigma}^2_\epsilon$. We choose a grid of values for ${\sigma}^2_\epsilon$ so that the \textit{sample (and population) standard deviation} of $\y$ range from 1 to 2, with evenly spaced increments of 0.10. The observed data is generated so that the \textit{sample (and population) mean} is fixed at 1 for all simulated data sets.\footnote{Since all BSL observes for the purpose of inference on $\theta$ is the sample mean and variance, generating data so that the sample mean and variance take on fixed values allows us to isolate the impact of model misspecification.}  The sample size across the experiments is taken to be $n=100$. Our prior for $\theta$ is $\theta\sim \mathcal{N}(0,10)$.

For BSL, we use $m=100$ simulated data sets to estimate the mean and variance of the summaries. Sampling of the BSL posterior is {implemented using random-walk Metropolis-Hastings (RWMH),} where the variance of the proposal is set to the exact posterior variance. The sampler is initialized at $\theta=0$, and run for 25,000 iterations, with the first 10,000 iterations discarded for burn-in.

Panel A of Figure \ref{fig_bsl1} plots the acceptance rates across the different levels of model misspecification. From the results in Panel A, we see that at large levels of misspecification RWMH has a very difficult time exploring the parameter space. As a consequence, the MCMC chain stays stuck for long periods of time and the acceptance rates plummet from a peak of about $70\%$, to a low of less than $0.01\%$. In Panel B of Figure \ref{fig_bsl1} we see that the posterior median of BSL, and the resulting credible set, vary significantly as the level of model misspecification in the data increases. Consequently, for values of the sample standard deviation larger than about 1.60, statistical inference based on the BSL posterior is not reliable. 

The inaccuracy with which the BSL posterior is sampled is a consequence of the incompatibility for $\eta_2(\y)$: no matter the value of $\theta$, the summary statistic $\eta_2(\y)$ can not be matched by the simulated counterpart $\eta_2(\zb)$. Therefore, the actual value of the observed statistic $\eta_2(\y)$, for any value of $\theta$, will be  in the tails of the Gaussian approximation for the posterior, which are inherently much noisier than values in the central mass of the approximation. Subsequently, this leads to a very noisy acceptance step and causes the MCMC chain to stick, with the overall result being unreliable inference for $\theta$. 

\begin{figure}[htb]
	\centering
	\includegraphics[width=17cm, height=6cm]{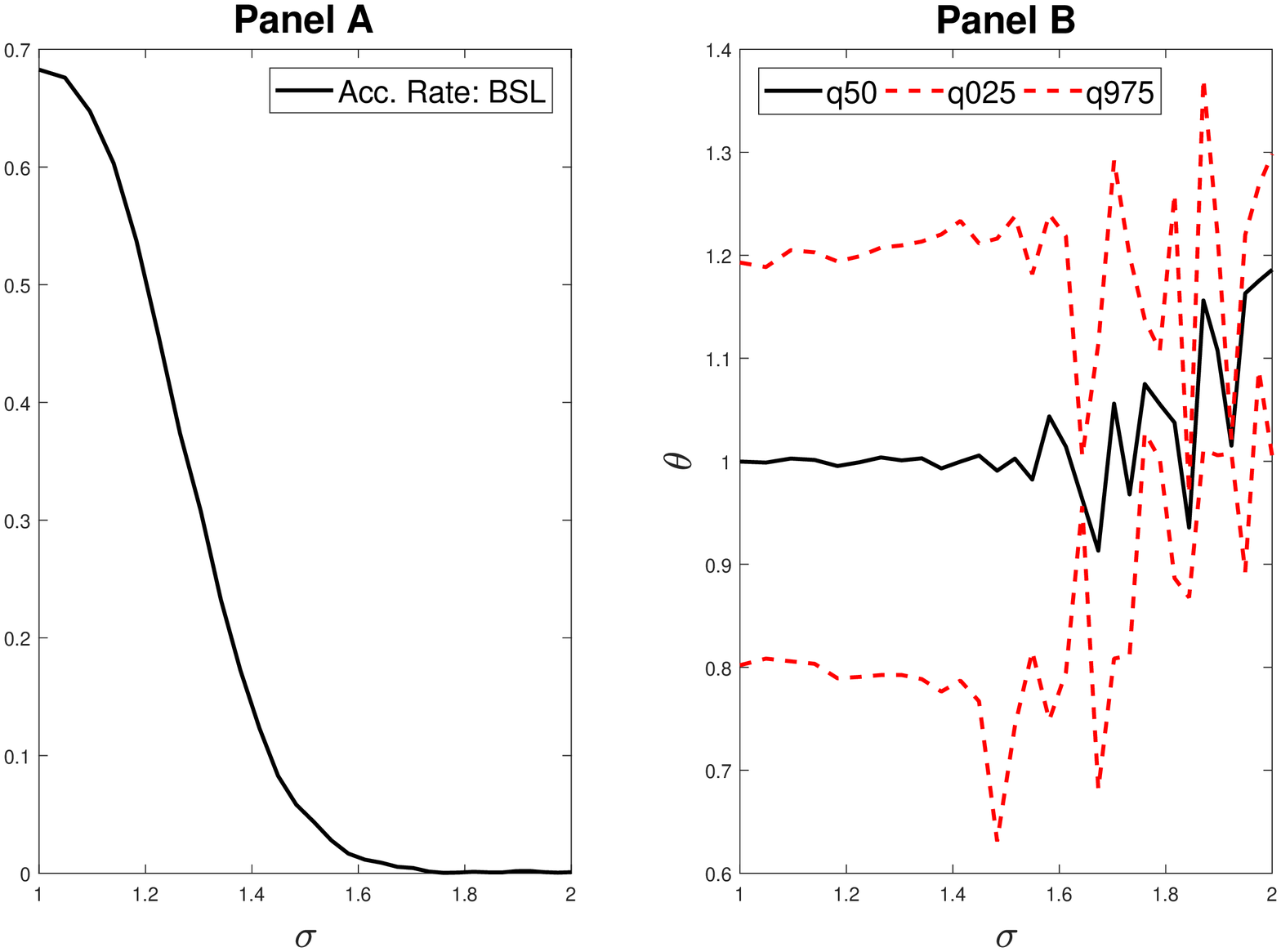}  
	\caption{{Panel A: acceptance rates for BSL across all levels of misspecification. Panel B: posterior medians and credible sets for BSL across the different levels of model misspecification. In both panels, the value of $\sigma$ denotes the sample (and population) standard deviation for the experiment. }}
	\label{fig_bsl1}
\end{figure}

\section{Robust Bayesian Synthetic Likelihood}
We propose two possible strategies for conducting inference using BSL when the model and summaries are incompatible {(i.e., when the model is misspecified in the BSL sense)}. The first approach augments the mean of the simulated summaries with additional free parameters, while the second approach augments the variance of the simulated summaries with additional free parameters. Both specifications allow us to conduct {reliable statistical inference on the model parameters regardless of model misspecification, and allows us to determine which of the summaries are incompatible with the data. Given this robustness,  throughout the remainder we refer to these approaches as ``mean'' and ``variance'' robust BSL (R-BSL).\footnote{We remind the reader that, following the robust statistics literature, we say a procedure is robust if it is not overly sensitive to model misspecification and/or if it can accurately diagnose model misspecification. Since this new approach accomplishes both of these tasks, the robust moniker is appropriate.}} 


\subsection{Mean Adjustment and Prior Specification}
Incompatibility implies that the observed statistics $\eta(\y)$ can not be recovered by the simulated mean $\mu_m(\theta)$, for any $\theta\in\Theta$, with probability converging to one.\footnote{Recall that $\eta(\y)$ is most often a sample mean so that its support is, with large probability, a shrinking ball around the point $b_0$.} Therefore, one approach to create a BSL procedure that is robust to this incompatibility issue is to adjust the vector of simulated means. This can be accomplished by adding to $\mu_m(\theta)$ an additional free parameter $\Gamma$, where $\Gamma=(\gamma_1,\dots,\gamma_{d_\eta})^\top\in\mathcal{G}\subset\mathbb{R}^{d_{\eta}}$, 
so that $\eta(\y)$ will always be in the support of this new simulated mean, even as the sample size diverges. Defining the joint vector of unknown parameters as $\zeta:=(\theta^{\top},\Gamma^\top)^\top\in\Theta\times \mathcal{G} \subset\mathbb{R}^{d_\theta}\times\mathbb{R}^{d_{\eta}}$, we define the vector of simulated means for use in BSL as
\begin{flalign*}
\phi_m(\zeta)&=\mu_m(\theta)+\text{diag}\left[\Sigma^{1/2}_m(\theta)\right]\Gamma.
\end{flalign*}Note that, by considering the scaled adjustment term, $\text{diag}[\Sigma^{1/2}_m(\theta)]\Gamma$, we ensure that these components are measured in the same units as $\mu_m(\theta)$, which allows us to treat $\Gamma$ as if they were unitless. 

Given this linear adjustment, and under weak conditions on the summary statistics and the parameter space $\Theta\times\mathcal{G}$, it is simple to see that $\phi_m(\zeta)$ will be compatible with $\eta(\y)$ for any prior choice on $\Gamma$ such that each individual component of $\Gamma$ has support over $\mathbb{R}$. 

Denote the prior on $\zeta$ by $\pi(\zeta)$. Following \cite{price2018bayesian}, the augmented BSL target, which we refer to as the Robust BSL-mean (R-BSL-M) posterior, is to generate samples from\footnote{The scaling of the perturbation by $\text{diag}[\Sigma^{1/2}_m(\theta)]$, gives the impression that the magnitude of the perturbation is decreasing as $n$ increases, since $\text{diag}[\Sigma^{1/2}_m(\theta)]$ is decreasing as $n$ increases. However, a simple exploration of the R-BSL-M posterior demonstrates that this scaling has no effect on these components, as they are themselves weighted by $\Sigma_m^{-1/2}(\theta)$ within the Gaussian kernel.} 
\begin{flalign*}
\hat{\pi}\left[\zeta\mid\eta(\y)\right]\propto \bar{g}_n\left[\eta(\y)|\zeta\right]\pi(\zeta),
\end{flalign*}where
\begin{flalign}\label{eq:rbsl1}
\bar{g}_n\left[\eta(\y)|\zeta\right]=\int \mathcal{N}\left[\eta(\y);\phi_m(\zeta),\Sigma_m(\theta)\right]\left\{\prod_{i=1}^{m}g_n[\eta(\zb^i)\mid\theta]\right\}\text{d}\eta(\zb^{1})\cdots\text{d}\eta(\zb^m).
\end{flalign}

\subsubsection*{Prior Choice: Laplace Prior}
To ensure that the observed summary $\eta(\y)$ can always be recovered by $\phi_m(\zeta)$, even when $n$ is very large, our prior on the components of $\Gamma$ should allow for $\text{diag}[\Sigma^{1/2}_m(\theta)]\Gamma$ to escape the support of $\mu_m(\theta)$ with large probability. However, given that some components of the original $\mu_m(\theta)$ are likely compatible with some components of $\eta(\y)$, we want to make sure that $\text{diag}[\Sigma^{1/2}_m(\theta)]\Gamma$ does not unduly perturb the components that are compatible. Therefore, we should choose a prior that places the vast majority of its mass near the origin. In this way, our prior choice for $\Gamma$ should induce \textit{``shrinkage''} in the components of $\Gamma$: only the components of $\Gamma$ that correspond to incompatible summaries should receive significant posterior probability away from the origin, while the components of $\Gamma$ that correspond to compatible summaries should have the majority of their posterior mass near the origin. 

With these dual requirements in mind,  and given that each component of $\Gamma$ has the same prior scale, we propose to follow the Bayesian Lasso literature (\citealp{park2008bayesian}) and use independent Laplace (i.e., double-exponential) priors for each component of $\Gamma$, with fixed location $0$ and common scale $\lambda>0$:
\begin{equation}\label{eq:laplace_prior}
\pi(\Gamma):=\prod_{j=1}^{d_\eta}\frac{1}{2\lambda} e^{-\frac{|\gamma_{j}|}{\lambda}}=\left(\frac{1}{2\lambda}\right)^{d_\eta}e^{-\frac{1}{\lambda}\sum_{j=1}^{d_\eta}|\gamma_j|}.
\end{equation} When convenient, we denote this prior by $\text{La}(0,\lambda)$. The Laplace prior for $\Gamma$ guarantees that the majority of prior mass for $\gamma_j$ is near the origin, but has thick enough tails so that $\phi_m(\zeta)$ is compatible with virtually any $\eta(\y)$ that would be used in practice. 

The hyper-parameter $\lambda$ should be chosen so that the prior support of $\pi(\Gamma)$ complements the support of $\mu_m(\theta)$. That is, $\lambda$ should be chosen so that the tails of $\Gamma$ are thicker than those of $\mu_m(\theta)$, which will allow us to detect deviations from compatibility, but not so large as to cause the statistic $\phi_m(\zeta)$ to have heavy tails. Indeed, if the tails of $\phi_m(\zeta)$ are too heavy, the implicit normality assumption made in BSL will be violated and can result in inefficient sampling. 

Since there is no reason to believe a priori that $\theta$ and $\Gamma$ are related, we take as our overall prior on $\zeta$ in R-BSL-M to be $$\pi(\zeta):=\pi(\theta)\text{La}(0,\lambda).$$  {As a default choice of prior, we suggest to select $\lambda = 0.5$, as this places most of the prior support for allowing up to $\pm 3$ standard deviations shift in the mean for each summary statistic (see Figure \ref{fig:robust_priors})}.

\begin{figure}[h!]
	\centering
	\includegraphics[width=15cm, height=7cm]{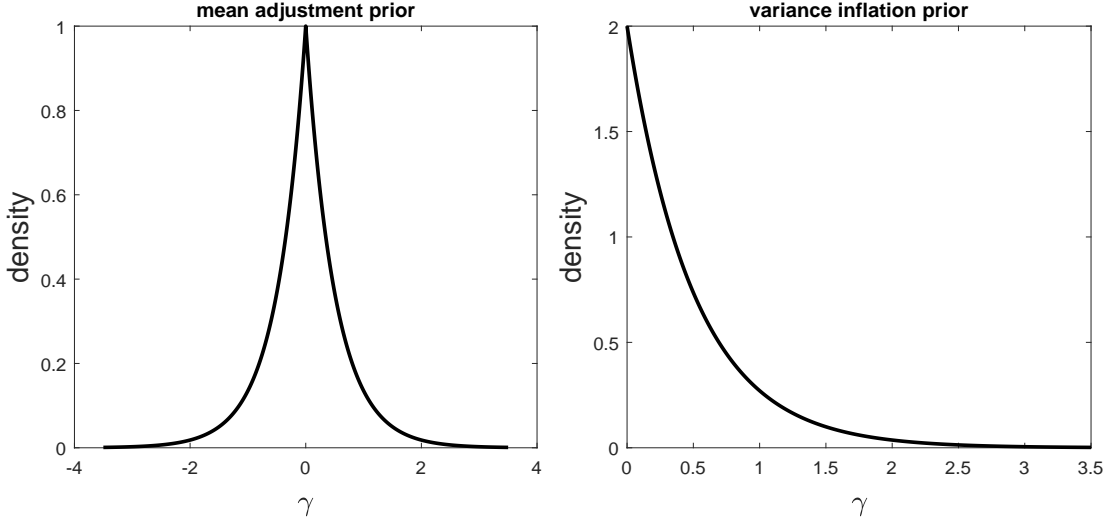}  
	\caption{Default prior distribution on $\gamma$ for mean adjustment (left) and variance inflation (right).}%
	\label{fig:robust_priors}
\end{figure}

\subsection{Variance Compatibility and Prior Specification}
While one approach to ensure compatibility is to adjust the mean of the simulated summaries, an alternative is to inflate the variance of the simulated summaries to ensure that $\eta(\y)$ is always in the support of $\mu_m(\theta)$.\footnote{Recall that, if $\eta(\y)$ is a sample average, the sample variance of $\mu_m(\theta)$ is converging to zero as either $n$ or $m$ diverges. Effectively, this has the effect of shrinking the support of $\mu_m(\theta)$.}

Under regularity conditions and for fixed $m$, it is likely to be the case that the centered statistic $Z_{n,m}(\theta):=\{\mu_m(\theta)-b(\theta)\}$ behaves as $Z_{n,m}(\theta)=O_{P}(1/(m\sqrt{n}))$, with the variance of $Z_{n,m}(\theta)$ decreasing like $1/(mn)$. Consequently, for $n$ (or $m$) large enough, if for a given value of $\theta$ the statistic $\{\eta(\y)-b(\theta)\}$ is more than a few standard deviations (as measured by $\Sigma^{1/2}_{m}(\theta)$) away from $\{\mu_m(\theta)-b(\theta)\}$, we can effectively view the summaries as being incompatible. 

Given this characterization, an alternative approach to ensure that $\mu_m(\theta)$ is compatible with $\eta(\y)$ is to artificially inflate the variance $\Sigma_m(\theta)$ so that the variance of $Z_{n,m}(\theta)$ never completely collapses to zero, and thus $\{\eta(\y)-b(\theta)\}$ can always be found in the support of $Z_{n,m}(\theta)$, albeit perhaps with small probability. More specifically, we propose to artificially inflate the variance used within BSL by adding to $\Sigma_m(\theta)$ the free parameters $\Gamma=(\gamma_1,\dots,\gamma_{d_\eta})^\top$. 

Recalling $\zeta=(\theta^\top,\Gamma^\top)^\top$, a robust BSL procedure based on adjusting the variance can be implemented by re-defining the variance of the simulated statistics used within BSL to be 
\begin{flalign}\label{eq:varadj}
V_{m}(\zeta):=\Sigma_{m}(\theta)+\begin{pmatrix}
[\Sigma_{m}(\theta)]_{11}^{}\gamma^2_1&0&\dots&0\\0&[\Sigma_{m}(\theta)]_{22}\gamma^2_2&\dots&0\\\vdots&\dots&\ddots&\vdots\\0&\cdots&\cdots&[\Sigma_{m}(\theta)]_{d_\eta d_\eta}^{}\gamma^2_{d_\eta}
\end{pmatrix},
\end{flalign}where $\left[\Sigma_m(\theta)\right]_{ii}$ denotes the $(i,i)$ element of $\Sigma_m(\theta)$. Given the structure of $V_m(\zeta)$, $\gamma_i$ can be interpreted as an inflation factor operating on the standard deviations of the original BSL variance.  An equivalent interpretation is that the $i$-th, $i=1,\dots,d_\eta$, BSL variance is multiplied by the factor $1+\gamma_i^2$. Using $V_m(\zeta)$ in place of $\Sigma_m(\theta)$ in the BSL posterior target, \eqref{eq:BSL_post}, and for $\pi(\zeta)$ an appropriate prior on $\zeta$, the Robust BSL-variance (R-BSL-V) posterior is given as: 
\begin{flalign*}
\hat{\pi}\left[\zeta|\eta(\y)\right]&\propto \bar{g}_n[\eta(\y)|\zeta]\pi(\zeta),
\end{flalign*}where
\begin{flalign}\label{eq:rbsl2}
\bar{g}_n\left[\eta(\y)|\zeta\right]=\int \mathcal{N}\left[\eta(\y);\mu_m(\theta),V_m(\zeta)\right]\left\{\prod_{i=1}^{m}g_n[\eta(\zb^i)\mid\theta]\right\}\text{d}\eta(\zb^{1})\cdots\text{d}\eta(\zb^m).
\end{flalign}

{We note that the variance adjustment approach is operationally similar to using a tempered version of the synthetic likelihood, where the parameter $\Gamma$ controls the ``amount'' of tempering, in that the value of $\Gamma$ allow us to artificially fatten the tails of the likelihood. Tempered likelihoods in Bayesian inference are often suggested as a means of conducting robust inference in the context of model misspecification (\citealp{bissiri2016general}). While interesting, a thorough comparison between these two approaches is beyond the scope of this paper and is left for future research. } 
		
\subsubsection*{Prior Choice: Exponential Prior}
Note that, by considering the standardization in \eqref{eq:varadj}, we ensure that each $\gamma_i$ has the same scale and can be considered as unit-less. Moreover, similar to the case of the mean adjustment BSL approach, there is no reason to believe there is any \textit{a priori} dependence between $\theta$ and $\Gamma$, so we can consider independent priors, i.e, $\pi(\zeta):=\pi(\theta) \pi(\Gamma)$. While several prior choices exist for $\Gamma$, following the arguments for the prior choice in the mean adjustment procedure, we need to choose a prior for the components of $\Gamma$ so that there is a large amount of prior mass near the origin, and enough mass out in the tails to ensure we can detect incompatible summaries. 

To this end, we consider independent exponential priors for each component $\gamma_i$, $(i=1,\dots,d_\eta)$, with common rate $\lambda>0$: 
\begin{flalign*}
\pi(\Gamma):=\prod_{i=1}^{d_\eta}\lambda e^{-\lambda \gamma_i}=\lambda^{d_\eta}e^{-\lambda\sum_{i=1}^{d_\eta}\gamma_i}.
\end{flalign*}
The hyper-parameter $\lambda$ should be chosen so that a large amount of prior mass is close to the origin, so as not to over-inflate the variance of the simulated summaries that are compatible. 

While this choice of prior is not, strictly speaking, a shrinkage prior, it is still the case that we should observe some shrinkage like behavior for summaries that are compatible. That is, for the summaries that are compatible, this additional inflation by $\Gamma$ is unnecessary and we expect that, for appropriate choices of $\lambda$, the addition of this component will not greatly affect the corresponding components in the variance. In contrast, for the summaries that are not compatible, this adjustment term is critical to ensure that the variance of the summaries is large enough to contain the observed summary $\eta(\y)$. {As a default choice of prior, we suggest to select a mean of $\lambda = 0.5$, as this places most of the prior support for allowing an additive inflation in the variance of up to $3$ times the standard deviation of each summary statistic (see Figure \ref{fig:robust_priors})}.

\subsection{Comparison of BSL and R-BSL}\label{sec:diffs}
At this stage, it is useful to compare and contrast BSL and our robust approach to BSL under model misspecification to understanding why R-BSL will produce more reliable/robust statistical inferences under model misspecification. To this end, we now compare the behavior of the dominant terms within the BSL and R-BSL likelihoods under model misspecification, which will highlight the fundamental differences between the posteriors that can emerge in practice.

In the case of R-BSL-M, the behavior of the posterior is driven by the quadratic form
 $$\|\phi_m(\zeta)-\eta(\y)\|_{\Sigma_m(\theta)}=\left[\phi_m(\zeta)-\eta(\y)\right]^{\top}{\Sigma^{-1}_m(\theta)}\left[\phi_m(\zeta)-\eta(\y)\right].$$ The R-BSL-M posterior assigns higher probability mass to values of $(\theta,\Gamma)$ for which $\|\phi_m(\zeta)-\eta(\y)\|_{\Sigma_m(\theta)}$ is ``small'', i.e., values of $(\theta,\Gamma)$ so that $\phi_m(\zeta)=\mu_m(\theta)+\text{diag}\left[\Sigma^{1/2}_m(\theta)\right]\Gamma$ is close to $\eta(\y)$. 
 
 The behavior of the R-BSL-V posterior is driven by the quadratic form
 $$ \|\mu_m(\theta)-\eta(\y)\|_{V_{m}(\zeta)}=\left[\mu_m(\theta)-\eta(\y)\right]^\top V_{m}^{-1}(\zeta)\left[\mu_m(\theta)-\eta(\y)\right].$$ If there are many values of $\theta$ for which $\|\mu_m(\theta)-\eta(\y)\|$ is already small, then variance inflation is not needed, and the resulting posterior for $\Gamma$ will be uninformative (and resemble the prior); if there are no values of $\theta$ that make $\|\mu_m(\theta)-\eta(\y)\|$ ``small'', then $\|\mu_m(\theta)-\eta(\y)\|_{V_{m}(\zeta)}$ can always be made small by choosing a large value of $\Gamma$. Consequently, when no value of $\theta$ exists for which $\mu_m(\theta)$ is close to $\eta(\y)$, the R-BSL-V posterior will assign high posterior mass to values of $\theta$ for which $\mu_m(\theta)$ is as close as possible to $\eta(\y)$, and values of $\Gamma$ that ensure $\|\mu_m(\theta)-\eta(\y)\|_{V_{m}(\zeta)}$ is small.

The posterior behavior of standard BSL is driven by the quadratic form
$$
\|\mu_m(\theta)-\eta(\y)\|_{\Sigma_m(\theta)}=\left[\mu_m(\theta)-\eta(\y)\right]^{\top}\Sigma_m^{-1}(\theta) \left[\mu_m(\theta)-\eta(\y)\right].
$$While R-BSL has in-built mechanisms to ensure that the these quadratic forms can be made small, no such mechanism exists for BSL: by the inequality $\|a-b-c\|\geq \|b\|-\|a\|-\|c\|$, 
\begin{flalign*}
\|\mu_m(\theta)-\eta(\y)\|_{\Sigma_m(\theta)}&\geq \|b(\theta)-b_0\|_{\Sigma_m(\theta)}-\|b(\theta)-\mu_m(\theta)\|_{\Sigma_m(\theta)}-\|\eta(\y)-b_0\|_{\Sigma_m(\theta)
}\\&\geq \|b(\theta)-b_0\|_{\Sigma_m(\theta)}-o_p(1),
\end{flalign*}
where the second inequality comes from the fact that $b_0:=\plim_n\eta(\y)$ and $b(\theta):=\plim_n\mu_m(\theta)$, for any $m\geq1$. Under incompatibility, i.e., model misspecification, the term $ \|b(\theta)-b_0\|_{\Sigma_m(\theta)}$ can be quite large, and is strictly positive in the limit, which results in a BSL posterior that is sensitive to the level/nature of model misspecification and which can ultimately be ill-behaved (for example, bi-modal); see Section \ref{sec:maexam} for a particular example. In contrast, since R-BSL ensures a form of compatibility can be achieved, the resulting posteriors will not be (particularly) sensitive to the level of model misspecification, and will be better behaved than their BSL counterpart.

\subsection{Sampling Robust BSL}

As demonstrated in \citet{price2018bayesian}, the standard BSL target posterior is given by
\begin{align}
\hat{\pi}[\zeta\mid\eta(\y)]&\propto\bar{g}_n[{\eta(\y)\mid\theta }]\pi(\theta), \label{eq:bsl_target}
\end{align}
where
\begin{align*}
\bar{g}_n[{\eta(\y)\mid\theta}] &= \int \mathcal{N}\left[\eta(\y);\mu_m(\theta),\Sigma_m(\theta)\right]\left\{\prod_{i=1}^{m}g_n[\eta(\zb^i)\mid\theta]\right\}\text{d}\eta(\zb^{1})\cdots\text{d}\eta(\zb^m).
\end{align*}
\citet{price2018bayesian} use a Metropolis-Hastings algorithm to sample from \eqref{eq:bsl_target} that proposes $\theta^* \sim q(\cdot|\theta)$ according to a Markov transition, and estimating $\bar{g}_n[{\eta(\y)\mid\theta^*}]$ unbiasedly through a single draw from $\prod_{i=1}^{m}g_n[\eta(\zb^i)\mid\theta^*]$ and evaluating $\mathcal{N}\left[\eta(\y);\mu_m(\theta),\Sigma_m(\theta)\right]$.   Using pseudo-marginal MCMC arguments of \citet{andrieu+r09}, substituting this estimator into a Metropolis-Hastings algorithm produces an algorithm that targets \eqref{eq:bsl_target}. 

Our robust BSL methods operate on an extended state space over $\theta$ and $\Gamma$ with target distribution
\begin{align*}
\hat{\pi}[\zeta\mid\eta(\y)]&\propto\bar{g}_n[{\eta(\y)\mid\zeta }]\pi(\theta)\pi(\Gamma), \text{ where }\\\bar{g}_n\left[\eta(\y)|\zeta\right]&=\int \mathcal{N}\left[\eta(\y);\phi_n(\zeta),\Sigma_n(\theta)\right]\left\{\prod_{i=1}^{m}g_n[\eta(\zb^i)\mid\theta]\right\}\text{d}\eta(\zb^{1})\cdots\text{d}\eta(\zb^m)
&\text{ (Mean Adjustment)},\\\bar{g}_n\left[\eta(\y)|\zeta\right]&=\int
\mathcal{N}\left[\eta(\y);\mu_m(\theta),V_m(\zeta)\right]\left\{\prod_{i=1}^{m}g_n[\eta(\zb^i)\mid\theta]\right\}\text{d}\eta(\zb^{1})\cdots\text{d}\eta(\zb^m)&\text{ (Variance Adjustment)}.
\end{align*}
To sample these target distributions we use a component-wise MCMC algorithm that updates, in turn, $\theta$ conditional on $\Gamma$ and then $\Gamma$ conditional on $\theta$.  The update for $\theta$ is the same as in standard BSL, but where the adjusted mean or inflated variance is computed as appropriate using the current value of $\Gamma$.  As before, the update for $\theta$ involves generating $m$ model simulations, $\eta(\zb^{i})$, $i=1,\ldots,m$.

The update for $\Gamma$ holds the currently accepted model simulations fixed, and thus $\mu_m(\theta)$ and $\Sigma_m(\theta)$ are fixed within the update step for $\Gamma$.  Each component of $\Gamma$, $\gamma_j$, for $j=1,\ldots,d_\eta$, is updated separately, conditional on the current values of the remaining components (denoted $\gamma_{/j}$).  The full conditional distribution for $\gamma_j$ is given by
\begin{align*}
\pi(\gamma_j^*|\theta,\mu_m(\theta),\Sigma_m(\theta),\gamma_{/j}) & \propto \mathcal{N}\left[\eta(\y);\phi_m(\zeta^*),\Sigma_m(\theta)\right]\pi(\gamma_j^*) &\text{(Mean Adjustment)} \\
\pi(\gamma_j^*|\theta,\mu_m(\theta),\Sigma_m(\theta),\gamma_{/j}) & \propto \mathcal{N}\left[\eta(\y);\mu_m(\theta),V_m(\zeta^*)\right]\pi(\gamma_j^*) &\text{(Variance Inflation)},
\end{align*} 
where $\gamma_j^*$ is a realisation of $\gamma_j$ and $\zeta^* = [\theta, \gamma_j^*, \gamma_{/j}]$.  

We sample this full conditional distribution using a slice sampler, and, in particular, the ``stepping out'' and ``shrinkage'' procedures detailed in \citet{Neal2003}.  The appeal of the slice sampler is that the acceptance probability is one, and thus there are no tuning parameters that can affect the statistical efficiency. However, a stepping out width needs to be selected, which can impact the speed of the slice sampler.  Given that the components of $\Gamma$ are eventually scaled by the summary statistic standard deviation, and given that our prior choices effectively penalise large values, we expect each component of $\Gamma$ to be ${O}(1)$.  Thus, we set the stepping out width to be 1, except for the lower bound of $\gamma_j$ in the variance inflation, which is immediately set to 0; and hence the stepping out procedure is not required.  We find this choice of width to be suitable, and since updating $\Gamma$ does not require any model simulations, the slice sampler is very fast.  Hence, importantly, our robust BSL methods do not require any additional tuning and the run-time per iteration for non-trivial applications is not noticeably slower.  The full MCMC algorithm to sample from the R-BSL posteriors is provided in Algorithm \ref{alg:robst_bsl}.

\begin{algorithm}[!htp]
	\SetKwInOut{Input}{Input}
	\SetKwInOut{Output}{Output}
	\Input{Summary statistic of the data, $\eta(\y)$, the prior distribution, $\pi(\theta)$, the proposal distribution $q$, the number of iterations, $T$, and the initial value of the chain, $\theta^{0}$, $\Gamma^0$.}
	\Output{MCMC sample $(\theta^{0},\theta^{1}, \ldots, \theta^{T})$ and $(\Gamma^0,\Gamma^1,\ldots,\Gamma^T)$ from the robust BSL posterior.  Some samples can be discarded as burn-in if required. }
	Estimate $\mu_m(\theta^{0})$ and $\Sigma_m(\theta^{0})$ via $m$ independent model simulations at $\theta^{0}$ \\
	Compute $\phi_m(\zeta^0)=\mu_m(\theta^0)+\text{diag}\left[\Sigma^{1/2}_m(\theta)\right]\Gamma^0$ (mean adjustment) or $V_{m}(\zeta^0)$ (variance inflation) defined in \eqref{eq:varadj}.   \\
	Compute robust synthetic likelihood $L^0 = \mathcal{N}\left[\eta(\y);\phi_m(\zeta^0),\Sigma_m(\theta^0)\right]$ (mean adjustment) or $L^0 = \mathcal{N}\left[\eta(\y);\mu_m(\theta^{0}),V_{m}(\zeta^0)\right]$ (variance inflation) \\
	\For{$i = 1$ to $T$}{
		\%\%\% Obtain $\Gamma^i$ via the following, which does not require any model simulations \\
		\For{$j = 1$ to $d_\eta$}{
			Update $\gamma_j^{i}$ with a slice sampler with target density $\pi(\gamma_j^{i}|\theta^{i-1},\mu_m(\theta^{i-1}),\Sigma_m(\theta^{i-1}),\gamma_1^{i}, \ldots, \gamma_{j-1}^{i},  \gamma_{j+1}^{i-1}, \ldots, \gamma_{d_\eta}^{i-1}) \propto \mathcal{N}\left[\eta(\y);\phi_m(\zeta^*),\Sigma_m(\theta^{i-1})\right]\pi(\gamma_j^{i})$ (mean adjustment) or $\mathcal{N}\left[\eta(\y);\mu_m(\theta^{i-1}),V_{m}(\zeta^*)\right]\pi(\gamma_j^{i})$ (variance inflation) where $\zeta^* = [\theta^{i-1}, \gamma_1^{i}, \ldots, \gamma_{j-1}^{i}, \gamma_j^{i},  \gamma_{j+1}^{i-1}, \ldots, \gamma_{d_\eta}^{i-1}]$  \\
		}
		Denote $L^i = \mathcal{N}\left[\eta(\y);\phi_m([\theta^{i-1}, \Gamma^i]),\Sigma_m(\theta^{i-1})\right]$ (mean adjustment) or $L^i = \mathcal{N}\left[\eta(\y);\mu_m(\theta^{i-1}),V_{m}([\theta^{i-1}, \Gamma^i])\right]$ (variance inflation)  \\
		\%\%\% Update $\theta^i$ conditional on $\Gamma^i$ \\
		Draw $\theta^{*} \sim q(\cdot|\theta^{i-1})$\\
		Estimate $\mu_m(\theta^*)$ and $\Sigma_m(\theta^*)$ via $m$ independent model simulations at $\theta^*$ \\
		Compute $\phi_m(\zeta^*)$ (mean adjustment) or $V_{m}(\zeta^*) $ (variance inflation) defined in \eqref{eq:varadj}, where $\zeta^* = [\theta^{*}, \Gamma^i]$  \\
		Compute proposed adjusted synthetic likelihood $L^* = \mathcal{N}\left[\eta(\y);\phi_m(\zeta^*),\Sigma_m(\theta^*)\right]$ (mean adjustment) or $L^* = \mathcal{N}\left[\eta(\y);\mu_m(\theta^*),V_{m}(\zeta^*)\right]$ (variance inflation)  \\
		Compute Metropolis-Hastings ratio:
		\begin{align*}
		r &= \frac{L^* \pi(\theta^*)q(\theta^{i-1}|\theta^*)}{L^i \pi(\theta^{i-1})q(\theta^*|\theta^{i-1})}
		\end{align*}
		\eIf{$\mathcal{U}(0,1)<r$}{
			Set $\theta^{i}=\theta^{*}$, $\mu_m(\theta^{i})=\mu_m(\theta^{*})$ and $\Sigma_m(\theta^{i})=\Sigma_m(\theta^{*})$\\
		}{
			Set $\theta^{i}=\theta^{i-1}$, $\mu_n(\theta^{i})=\mu_n(\theta^{i-1})$ and $\Sigma_n(\theta^{i})=\Sigma_n(\theta^{i-1})$\\
		}
		\caption{Robust MCMC BSL.}
		\label{alg:robst_bsl}
	}
\end{algorithm}

It is important to note that under either R-BSL approach, we recover the original BSL target when $$\gamma_1=\dots=\gamma_{d_\eta}=0.$$Therefore, if the model $P^{n}_\theta$ can generate summaries $\eta(\zb)$ that match $\eta(\y)$, under an appropriate prior specification, the posterior $\hat{\pi}[\Gamma\mid\eta(\y)]$ should not differ substantially from the prior, with most of the posterior mass located near the origin. 

\subsection{Theoretical Properties of R-BSL}
If the model is compatible, what behavior should we expect from the R-BSL approach? Given the priors used in R-BSL, which place the majority of their mass near the origin, one would hope that when the compatibility condition (Definition \ref{def1}) is satisfied, the introduction of the additional parameters $\Gamma$ does not influence the BSL posterior for $\theta$. We demonstrate that when Definition \ref{def1} is satisfied the asymptotic behavior of the R-BSL posterior {of $\theta$ components behaves the same as the exact BSL posterior: namely, both posteriors concentrate all mass onto the value of $\theta$ that satisfies $b(\theta)=b_0$. In addition, we demonstrate that when Definition \ref{def1} is satisfied the R-BSL posteriors for the adjustment components converge to the prior. Consequently, we do not lose anything by using R-BSL when the model is correctly specified, but gain robustness to deviations from the modeling assumptions if the model is misspecified.} 


Before presenting the formal result, we must state some notation. Define $P^n_{0}$ as the true distribution generating $\mathbf{y}$. The map $\eta:\mathbb{R}^n\rightarrow\mathbb{R}^{d_\eta}$, $d_\eta\ge d_\theta$, which defines the summary statistics used in the procedure, satisfies $\eta(\y)\sim G_n^0$, where $G^0_n$ denotes the projection of $P_n^0$ under the map $\eta$, and denote by $g_n^0$ the density of $G_n^0.$ Likewise, recall that $G_n[\cdot|\theta]$ denotes the projection of $P^\theta_n$ under the map $\eta$. For real-valued sequences $\{a_{n}\}_{n\geq 1}$ and
$\{b_{n}\}_{n\geq 1}$: $a_{n}\lesssim b_{n}$ denotes $a_{n}\leq Cb_{n}$ for
some finite $C>0$ and all $n$ large, $a_{n}\gtrsim b_{n}$ denotes $a_{n}\geq Cb_{n}$ for
some finite $C>0$, and $a_{n}\asymp b_{n}$ implies $a_{n}\lesssim b_{n}$ and $a_{n}\gtrsim b_{n}$. The terms $O_P$ and $o_P$ have their usual connotations and the notation $\Rightarrow$ denotes weak convergence in distribution. 

Recall that, $b_{}(\theta)=\mathbb{E}[\eta(\y)|\theta]$ and $\Sigma(\theta):=\mathbb{E}\left[\left\{\eta(\y)-b(\theta)\right\}\left\{\eta(\y)-b(\theta)\right\}'\mid \theta\right].$ We impose the following regularity conditions. 
\begin{assumption}\label{ass:one}{
		There exists a sequence of positive real numbers $v_n$ diverging to $\infty$ such that, for some distribution $Q$ on $\mathbb{R}^{d_\eta}$ and some vector $b_0\in\mathbb{R}^{d_\eta}$, $v_n\left[\eta(\y)-b_0\right]\Rightarrow Q,\text{ \em{under} } P^0_n.$}
\end{assumption}
\begin{assumption}\label{ass:two}
	(i) The sequence $\{v_n\}_{n\ge1}$ is such that, for all $\theta\in\Theta$ and some $n$ large enough, there exists constants ${c}_1,{c}_2$, $c_1\leq c_2$, satisfying: $0<{c}_1\leq \|v_{n}^{}{\Sigma}^{}_{n}({\theta }%
	)\|_*\leq {c}_2<\infty$, for some matrix norm $\|\cdot\|_*$; (ii) For all $\theta\in\Theta$ and some $n$ large enough, the $d\times d$-matrix $A_{n}(\theta)$ is continuous in $\theta$.
\end{assumption}
\begin{assumption}\label{ass:three}
There exists a deterministic map $\theta\mapsto{b}(\theta)$, such that, for all $\theta\in\Theta$, and for constants $\alpha, u_0>0$, for all $0<u<u_0 v_n$,
${G}_n\left[ \|v_n\{\eta(\y)-b(\theta)\}\|>u\mid\theta\right] \leq c(\theta) u^{-\alpha},$
uniformly for $n\ge1$ and where $\int_\Theta c(\theta)\pi(\theta)\text{d}\theta=O(1). $
\end{assumption}

\begin{assumption}\label{ass:four}(i) There exists some $\tau>0$ such that, for all $0<u <u_0 v_n$, the prior probability satisfies
	$\Pi \left[ \|{b}({\theta })-{b}_{0}\|\leq u \right]
	\asymp u ^{\tau}.$ (ii) The prior density $\pi(\theta)$ is continuous and satisfies $\pi(\theta_0)>0$. 
\end{assumption}

\begin{assumption}\label{ass:five} {(i)} The map $\theta\mapsto{b}(\theta)$
	is continuous and injective, with $b(\theta_0)=b_0$ for some $\theta_0\in\Theta$, and satisfies:
	$\Vert {\theta }-{\theta }_{0}\Vert \leq L\Vert {b}({\theta })-b_0\Vert ^{\kappa }$ on some open neighbourhood of ${\theta }_{0}$ with $L>0$ and $\kappa >0$.
\end{assumption}

\begin{assumption}\label{ass:six} If Assumption \ref{ass:five} is satisfied, for any $\epsilon>0$, there exists $u,\delta>0$ and a set $V_n$ such that, for all $\theta\in\{\theta:\|b(\theta)-b_0\|\leq u v_n^{-1}\}$, $V_n\subset \left\{\eta\in\mathbb{R}^{d_\eta}:g_n^0(\eta)\lesssim {g}_n\left(\eta\mid \theta\right)\right\}\text{ where } P_n^0(V_n^c)<\epsilon.$
	
\end{assumption}
The above assumptions are similar to those used in \cite{marin2014relevant} to deduce the behavior of Bayes factors in situations where inference is conditioned on summary statistics, as opposed to the entire data set. Due to space constraints, we defer a detailed discussion on these assumptions to Section 2 of the supplementary material. 

{The following result, the proof of which is also given in the supplemental material, describes the theoretical behavior of the R-BSL posterior under the above assumptions.} 

\begin{proposition}\label{prop:one} 
Under Assumption \ref{ass:one}-\ref{ass:six}, for any $\delta>0$, 
$$
\Pi\left[\|\theta-\theta_0\|\leq\delta |\eta(\y)\right]=1+o_P(1).
$$
Moreover, for any $A\subseteq\mathcal{G}$: $$\Pi\left[\Gamma\in A|\eta(\y)\right]=\Pi[\Gamma\in A]+o_P(1).$$
\end{proposition}
	Proposition \ref{prop:one} demonstrates that under compatibility the posterior for the model parameters, $\theta$, are asymptotically unaffected by the introduction of the $\Gamma$ components: the R-BSL posterior for $\theta$ concentrates all posterior mass onto the value $\theta_0$, the value under which $b(\theta)=b_0$. {Consequently, from Proposition 1 in \cite{frazier2019bayesian}, the R-BSL and BSL posteriors behavior similarly when the model is compatible (i.e., when Definition \ref{def1} is satisfied) and implies that R-BSL does not pay a penalty for being robust to model misspecification if the model is correctly specified.\footnote{We note that determining the theoretical behavior of BSL and R-BSL when compatibility is not in evidence is a significant undertaking, and a novel research topic in its own right. Therefore, we leave such technical details for future research.}}

	Proposition \ref{prop:one} also demonstrates that, under compatibility, the posterior for the $\Gamma$ components converge to the prior. This implies that under correct model specification the resulting posteriors for the $\Gamma$ components will not asymptotically concentrate on the origin. As discussed in \cite{bhattacharya2012bayesian}, this is not surprising given the relatively mild shrinkage priors placed on the adjustment components. We conjecture that if stronger shrinkage priors were employed, posterior concentration, toward the origin, for these components could also be achieved. However, the use of these more complex priors could create issues within the sampling. 
			
	When the compatibility condition in Definition \ref{def1} is not satisfied, the R-BSL posterior for $\Gamma$ will deviate from the prior. Therefore, the result of Proposition \ref{prop:one} can be used to determine the level of model misspecification by comparing the difference between the R-BSL posterior for $\Gamma$ and the prior for $\Gamma$. While visual detection will often be enough to determine if any meaningful differences between these two exist, any norm on the space of probability measures could be used to quantify this discrepancy.

	By analyzing the posterior elements of $\Gamma$ that differ from the prior, we can deduce precisely which of the summary statistics the assumed model can not match, i.e., which summaries are not compatible. This information can then be incorporated into subsequent modelling steps to construct a model that can more adequately capture the observed data. In this way, R-BSL can also be used as a model criticism device to help researchers locate discrepancies between the assumed model and the observed data.

\section{Examples}\label{sec:examples}
{In this section, we first consider two toy examples that demonstrate the statistical benefits of R-BSL, over BSL.} Next, we apply BSL and R-BSL to conduct inference in a challenging problem in ecology: using real data to model invasive toad populations. {Collectively, these examples echo the analysis in Section \ref{sec:diffs} and demonstrate that R-BSL delivers more robust statistical inferences than BSL under model misspecification.} {In all examples we use the default priors shown in Figure \ref{fig:robust_priors}.}

\subsection{Toy Example 1 Continued: Contaminated Normal Model}\label{sec:toynormcont}

In this section, we analyze the performance of R-BSL in the contaminated normal example. We refer the reader to Section \ref{sec:toynorm} for a detailed description of the underlying Monte Carlo design. 
Following the analysis in Section \ref{sec:toynorm}, we use precisely the same ``observed'' data in Section \ref{sec:toynorm}, and apply R-BSL to this data.\footnote{We again use a RWMH algorithm that is initialized at  $\theta=0$, and we run the sampler for 25,000 iterations, with the first 10,000 discarded for burn-in.} The acceptance rates for BSL, and R-BSL are plotted in Panel A of  Figure  \ref{fig:normal_rbsl}, while the resulting posterior median and credible sets are given in Panels B, C and D. 

{The results in Panels B, C and D of Figure \ref{fig:normal_rbsl} demonstrate that R-BSL yields reliable statistical inferences on $\theta$ regardless of the level of model misspecification. In stark contrast to BSL, across all the experiments, the posterior means and credible sets for R-BSL are virtually unaffected by the data contamination. Panel A of Figure  \ref{fig:normal_rbsl} demonstrates that the acceptance rate for the variance adjustment R-BSL approach is only slightly affected by the level of model misspecification, while the mean adjustment version does display some degradation but still maintains acceptance rates above $5\%$ in all cases. In comparison, when the sample standard deviation is equal to 2, the BSL acceptance rate is less than $0.01\%$.}

{In addition, we note that numerically implementing R-BSL is not much more computationally costly than implementing BSL. In this example, the execution time required to sample the R-BSL posterior (via Algorithm \ref{alg:robst_bsl}) was only 40\% slower than the execution time required to sample the BSL posterior, even though R-BSL is conducting inference on three times as many parameters. Moreover, the additional time per MCMC iteration for R-BSL is even less noticeable in realistic examples where model simulation is non-trivial.} 
\begin{figure}[h!]
	\centering
	\includegraphics[width=17cm, height=7cm]{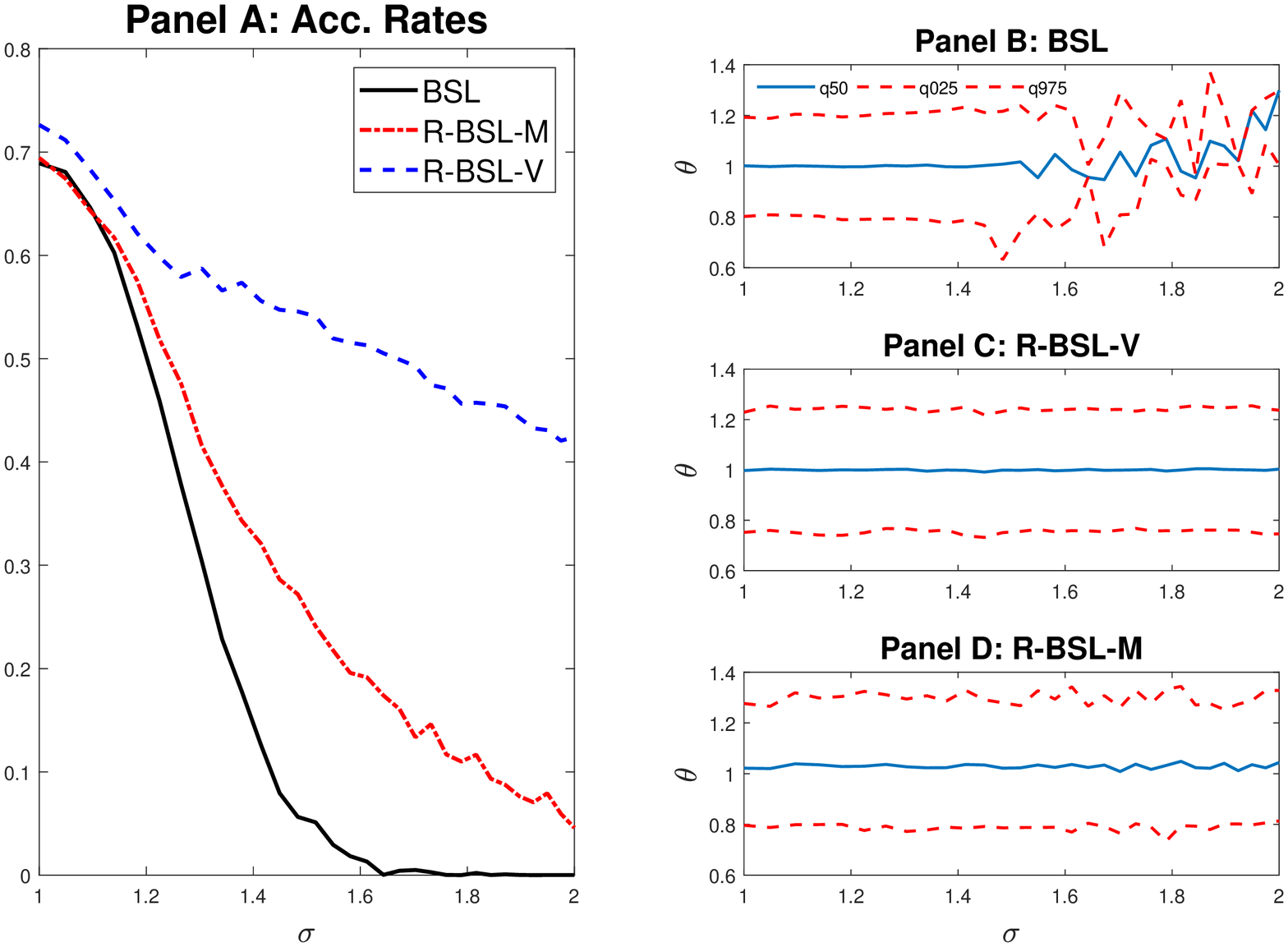}  
	\caption{\footnotesize{Panel A gives the acceptance rates for the different BSL procedures across the different data sets. Panels B, C and D give the posterior median and 95\% credible sets for R-BSL and compares these with those obtained from BSL. We recall that the value of $\sigma$ denotes the value of the sample (and population) standard deviation.}}%
	\label{fig:normal_rbsl}
\end{figure}

In the supplementary material, we analyze the adjustment components for this example and compare the repeated sampling behavior of BSL and R-BSL across different levels of model misspecification. As expected, the R-BSL posteriors for the adjustment components demonstrate that the model is unable to match the second summary statistic, while the posteriors associated with the first adjustment component are indistinguishable from the prior. See Figure 1 in the supplementary appendix for full details. The repeated sampling results demonstrate that R-BSL behaves similarly to BSL when the model is correctly specified, but yields more accurate estimators when the model is misspecified. See Table 1 in the supplementary appendix for full details. 

\subsection{Toy Example 2: Moving Average Model}\label{sec:maexam}
A common toy example used to demonstrate approximate inference methodology is the moving average (MA) model. The researcher believes $\mathbf{y}$ is generated according to an MA(1) model:
\begin{equation}
z_{t}=e_{t}+\theta _{}e_{t-1},  \label{MA2}
\end{equation}and the unknown parameter $\theta$ satisfies $|\theta|<1$, while our prior information on $\theta$ uniform over $(-1,1)$. A useful choice of summary statistics for the MA(1) model are the
sample autocovariances $\eta _{j}(\mathbf{z})=\frac{1}{T}%
\sum_{t=1+j}^{T}z_{t}z_{t-j}$, for $j\in\{0,1,2\}$. Let $\eta(\mathbf{z})$ denote the summaries $\mathbf{\eta}\left( \mathbf{z}\right) =(\eta _{0}\mathbf{(z)}, \eta _{1}\mathbf{(z)}, \eta _{2}\mathbf{(z)})^\top$. Under the DGP in equation \eqref{MA2}, it can be shown that the summaries $\eta(\zb)$ satisfy 
$$
\eta(\zb)\xrightarrow{P}b(\theta):=\begin{pmatrix}
1+\theta^2_{1},&\theta_{1},&0
\end{pmatrix}^{\top}.$$

While the researcher believes the data is generated according to an MA(1) model, the actual DGP for $\mathbf{y}$ evolves according to the stochastic volatility (SV) model  
\begin{flalign}
y_{t}=\exp(h_{t}/2)u_{t},\;\;h_{t}=\omega+\rho h_{t-1}+v_{t}\sigma_{v}\label{trueDGP},
\end{flalign} where $0<\rho<1$, $0<\sigma_{v}<1$, $u_{t}$ and  $v_{t}$ and both iid standard normal. In this case, if one takes $\mathbf{\eta }\left( \mathbf{y}\right) =(\eta _{0}\mathbf{(y)}, \eta_{1}\mathbf{(y)}, \eta _{2}\mathbf{(y)})^\top$, under the DGP in \eqref{trueDGP}, 
\begin{equation*}
\eta(\mathbf{y})\stackrel{P}{\rightarrow} b_{0}:=\begin{pmatrix}\exp\left( \frac{\omega}{1-\rho}+\frac{1}{2}\frac{\sigma_v^2}{1-\rho^2}\right)
,&0,&0
\end{pmatrix}^{\top}.\end{equation*}

For any value of $\omega,\sigma_v$ and $\rho$ such that $\exp\{ {\omega}/{(1-\rho)}+\frac{1}{2}{\sigma_v^2}/{(1-\rho^2)}\}\neq1,$ the model is not compatible (and hence misspecified in the BSL sense); i.e., for any value of $\theta$, $\|b(\theta)-b_0\|>0$. From the definition of $b(\theta)$ and $b_0$, it also follows that the value that minimizes $\|b(\theta)-b_0\|$ is $\theta=0$, and it is this value onto which we would expect the R-BSL posterior to concentrate.

To understand how BSL and R-BSL perform in this misspecified model, we enact the following Monte Carlo experiment: we generate $n=100$ observations from the SV model in \eqref{trueDGP} with parameter values $\omega = -0.76$, $\rho = 0.90$ and $\sigma_v = 0.36$, and use BSL, and R-BSL to conduct inference on $\theta$ in the misspecified MA(1) model. The R-BSL approach uses $m=50$ simulated data sets to estimate the mean and the variance, while BSL uses $10\times m=500$ simulated data sets. 

Under the true DGP in \eqref{trueDGP}, the first two auto-correlations are zero for all values of $(\omega,\rho,\sigma_v)^\top$. Hence, we would expect that R-BSL will detect incompatibility in the first summary statistic, the sample variance, while the corresponding adjustment components for the other summaries would be indistinguishable from the prior, given in Figure \ref{fig:robust_priors}. 

For R-BSL we consider starting values obtained from the maximum likelihood estimators of the MA(1) model, and posterior draws for $\theta$ are obtained via a random-walk Metropolis sampler with fixed variance of $0.1$. We run the MCMC sampler for 100,000 iterations and discard the first 10,000 for burn-in. 

{The nature of the model misspecification in this example results in a standard BSL posterior that is ill-behaved; i.e., it is bi-modal with well separated modes, and the MCMC struggles to move between modes.  Thus, given that there is only a single parameter, we use importance sampling based on 100,000 samples from the prior for $\theta$ to obtain the BSL posterior.  The effective sample sizes of the importance sampling approximations over the 50 datasets is roughly between 50--500, allowing us to reasonably estimate the standard BSL posterior.}

Under this Monte Carlo design, we generate fifty replications from the DGP in \eqref{trueDGP} and apply BSL and R-BSL to the `observed data'. The average acceptance rates across the replicated data sets for R-BSL-V and R-BSL-M are $38\%$ and $26\%$, respectively. We display the posteriors for $\theta$ from each method, and across each data set, in Panels A-C of Figure \ref{fig:ma2_theta}.\footnote{To deal with excess autocorrelation in the MCMC chains, the results in Panel A and B of Figure \ref{fig:ma2_theta} have been thinned by taking every 100th sample. This significantly reduced the autocorrelation in the chain and permits smoother density estimates. Furthermore, we note here that the MCMC chains associated with each of the, non-thinned, posteriors satisfies the convergence diagnostic proposed in \cite{Geweke1992} at any reasonable level of significance.} {Both R-BSL-V and R-BSL-M display significant posterior concentration around $\theta=0$, while the BSL posterior is bi-modal and has little posterior mass around $\theta=0$. The differences between the posteriors highlights the results discussed in Section \ref{sec:diffs}, where we argued that the BSL and R-BSL posteriors can be dissimilar when the model is misspecified. }

\begin{figure}[h!]
	\centering
	\includegraphics[width=18.5cm, height=7cm]{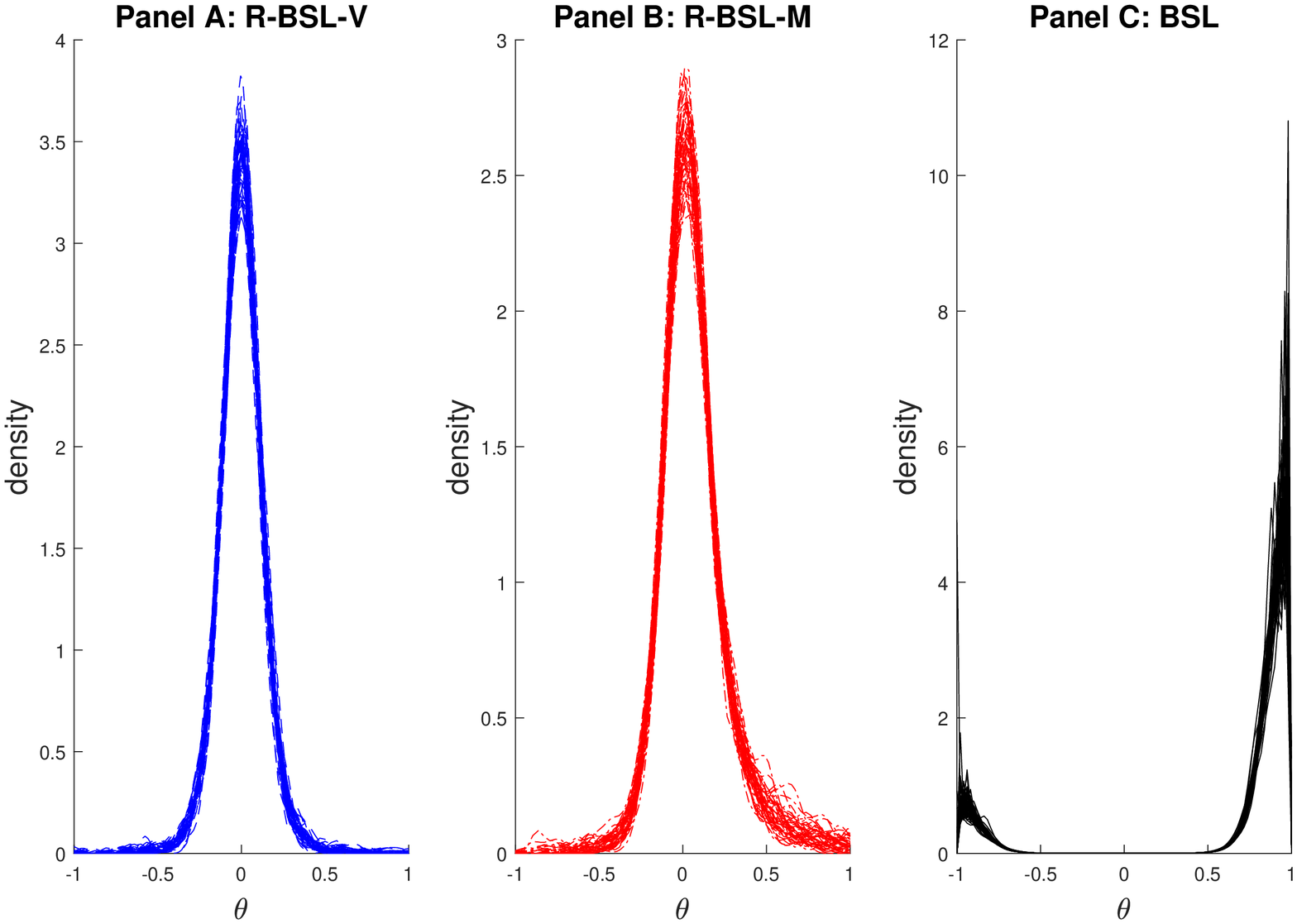}  
	\caption{\footnotesize{Posteriors for BSL, R-BSL-M and R-BSL-V for $\theta$ in the misspecified MA(1) model across fifty replicated data sets.}}%
	\label{fig:ma2_theta}
\end{figure}

We now examine the marginal posteriors for the different adjustment components across R-BSL-V and R-BSL-M, which are given in Figure \ref{fig:ma2_adjust}, across the fifty replications. The top row of Figure \ref{fig:ma2_adjust} gives the posterior densities of $\gamma_1$, $\gamma_2$ and $\gamma_3$ for R-BSL-M and the bottom row corresponds to the same components for R-BSL-V. The results demonstrate that, as suggested, the model can not reliably match the first summary statistic, the sample variance, while the posteriors corresponding to the first and second-order autocorrelations do not significantly differ from their priors across all replicated data sets. 
\begin{figure}[h!]
	\centering
	\includegraphics[width=16cm, height=9cm]{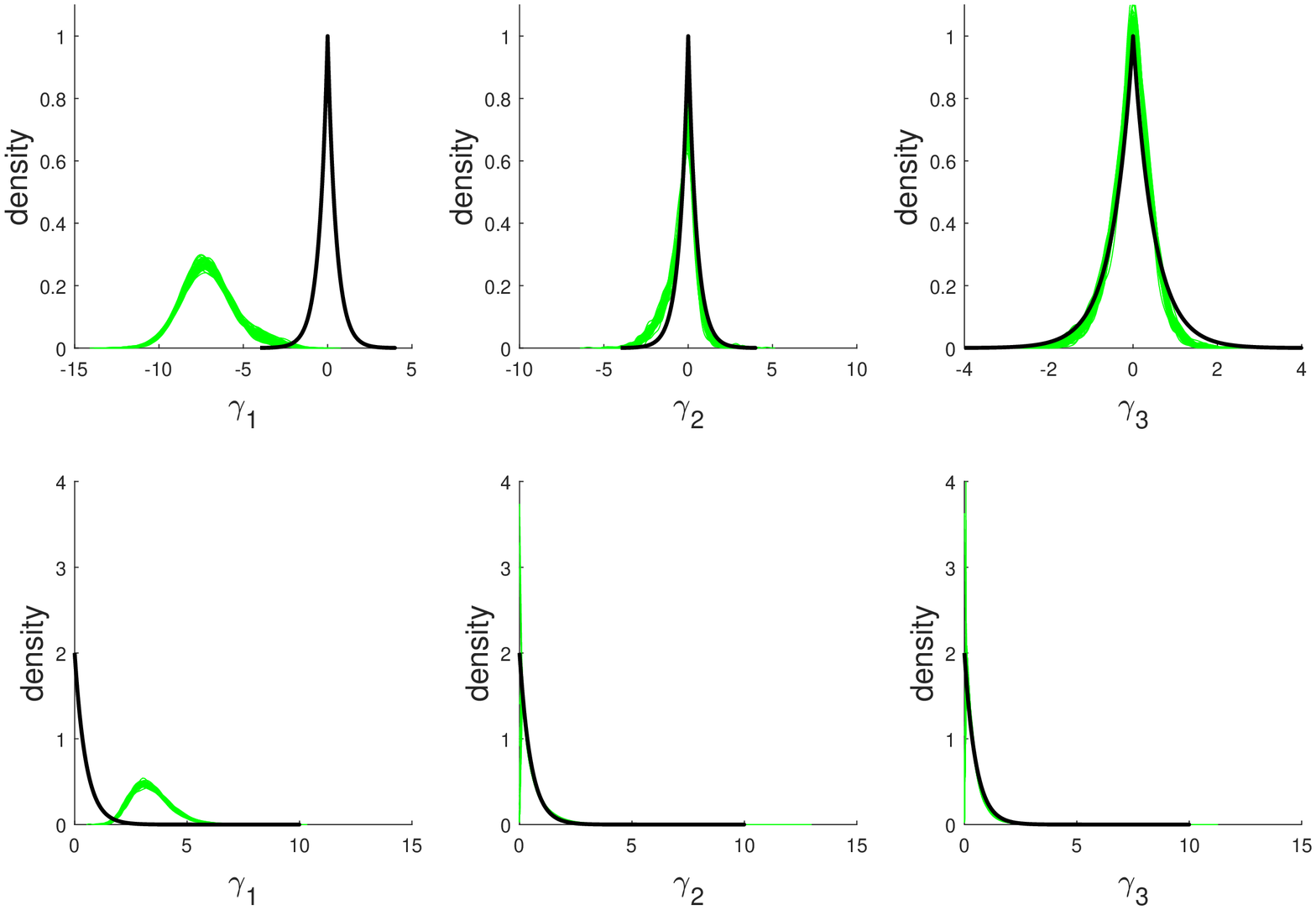} 
	\caption{\footnotesize{The top panels plot the R-BSL-M posteriors for $\gamma_{1},\gamma_2,\gamma_3$, thin lines, and the corresponding priors, thick lines, across the fifty replicated data sets. The bottom panels displays the same information for R-BSL-V.} }%
	\label{fig:ma2_adjust}
\end{figure}

The accuracy of the BSL and R-BSL point estimators, across the repeated samples, is analyzed in Table \ref{tab:ma2}. For each data set we run BSL, R-BSL-M, and R-BSL-V and calculate the bias (Bias), root mean squared error (RMSE), posterior credible set length (Len), and Monte Carlo coverage (COV), all relative to the pseudo-true value $\theta=0$. The results are displayed in Table \ref{tab:ma2}, and demonstrate that R-BSL-V yields the most accurate point estimators, as measured by both bias and RMSE, followed by R-BSL-M. The bi-modal nature of the BSL posterior leads to a significantly biased point estimator, and ensures that the resulting credible set length and Monte Carlo coverage are not entirely meaningful. Hence, we do not report these quantities for BSL in Table \ref{tab:ma2}. 

The results in Figure \ref{fig:ma2_theta} and Table \ref{tab:ma2} unequivocally demonstrate that when the model is misspecified R-BSL yields more reliable statistical inferences than those obtained by BSL.

\begin{table}[h!]
	\centering
	\centering
	\caption{ \footnotesize{Summary measures for posterior mean accuracy, calculated as averages across the replications. RMSE- root mean squared error, BIAS- bias across the replications, LEN- credible set length, COV- Monte Carlo coverage. {Due to the bi-modal BSL posterior, we do not report the resulting Monte Carlo coverage (COV) or credible set length (LEN) for BSL.}}}
	\begin{tabular}{lrrrr}
		    & \multicolumn{1}{l}{BSL} & \multicolumn{1}{l}{R-BSL-V} & \multicolumn{1}{l}{R-BSL-M} &        \\
		RMSE  & 0.306 & 0.006 & 0.076 &       \\
		BIAS  & 0.305 & -0.001 & 0.073 &       \\
		LEN   & N/A& 0.544 & 0.979 &       \\
		COV   & N/A  & 100\% & 100\% &       \\
	\end{tabular}%
	\label{tab:ma2}%
\end{table}
\subsection{Toad Example}\label{sec:toad}

\subsubsection{Background}

We consider an individual-based model of a species called Fowler's Toads (\textit{Anaxyrus fowleri}) developed by \citet{Marchand2017}, which was also analysed by \citet{an2018robust}.  Here we give very brief details, with more information in \citet{Marchand2017} and \citet{an2018robust}.

The model assumes that a toad hides in its refuge site in the daytime and moves to a randomly chosen foraging place at night.  GPS location data are collected on $n_t$ toads for $n_d$ days, i.e.\ the observation matrix $\vect{Y}$ is of dimension $n_d \times n_t$ ($n_t=66$ and $n_d=63$ here).  Then $\vect{Y}$ is summarised to $4$ sets comprising the relative moving distances for time lags of $1,2,4,8$ days. For instance, $\vect{y}_1$ consists of the displacement information of lag $1$ day, $\vect{y}_1 = \{|\vect{Y}_{i,j}-\vect{Y}_{i+1,j}| ; 1 \leq i \leq n_{d}-1, 1 \leq j \leq n_t \}$.

Simulating from the model involves two distinct processes. For each toad, we first generate an overnight displacement, $\Delta y$, then mimic the returning behaviour with a simplified model. The overnight displacement is assumed to belong to the L\'evy-alpha stable distribution family, with stability parameter $\alpha$ and scale parameter $\delta$.  The total returning probability is a constant $p_0$, if a return occurs on day $i$, $1 \leq i \leq m$, then the return site is the same as the refuge site on day $i$, where $i$ is selected randomly from ${1,2,\dots,m}$ with equal probability.  Here we consider both simulated and real datasets. For the synthetically generated data we take $\theta=(\alpha,\delta,p_0)^\top=(1.8,45,0.6)^\top$, which is informed by the parameter estimates obtained in \citet{Marchand2017}. We use a uniform prior over $(1,2) \times (0,100) \times (0,0.9)$.  \citet{Marchand2017} consider three variations on the model.  Here, we consider their `Model 2' since there is a strong indication from their results that this model is not able to recover some of the chosen summary statistics.

As in \citet{Marchand2017}, the dataset of displacements is split into two components.  If the absolute value of the displacement is less than 10 metres, it is assumed the toad has returned to its starting location.  For the summary statistic, we consider the number of toads that returned (\citet{Marchand2017}).  For the non-returns (absolute displacement greater than 10 metres) we consider a larger collection of summaries.  We calculate the log difference between adjacent $p$-quantiles with $p=0,0.1,\ldots,1$ and also the median.  These statistics are computed separately for the four time lags. This results in 48 statistics in total, which is hard to handle for conventional ABC methods.  \citet{an2018robust} demonstrate that BSL is computationally efficient enough to analyse simulated data for this application with a similar number of summary statistics.

For the $\gamma$ parameters of the mean adjustment and variation inflation procedures we use the same priors as the previous example.

\subsubsection{Results}

We first consider the simulated dataset, where we use $n=300$ simulations to estimate the synthetic likelihood at each MCMC iteration.  Standard BSL, together with the two incompatibility extensions, produce approximate posteriors shown in Figure \ref{fig:toad_model2_simdata_posteriors}.  As can be seen, the adjustments produce posteriors remarkably similar to BSL with slightly inflated variances.   The MCMC acceptance rates for BSL, R-BSL-M and R-BSL-V are 11\%, 9\% and 22\%, respectively.  As consistent with previous results, the variance adjustment improves the computational efficiency, even when the model is correctly specified.  The posterior distributions for $\gamma$ of R-BSL are shown in Figure \ref{fig:toad_model2_simdata_gamma}.  In all cases, the posteriors are not too dissimilar to the prior.

\begin{figure}[h!]
	\centering
	\includegraphics[width=16cm, height=5cm]{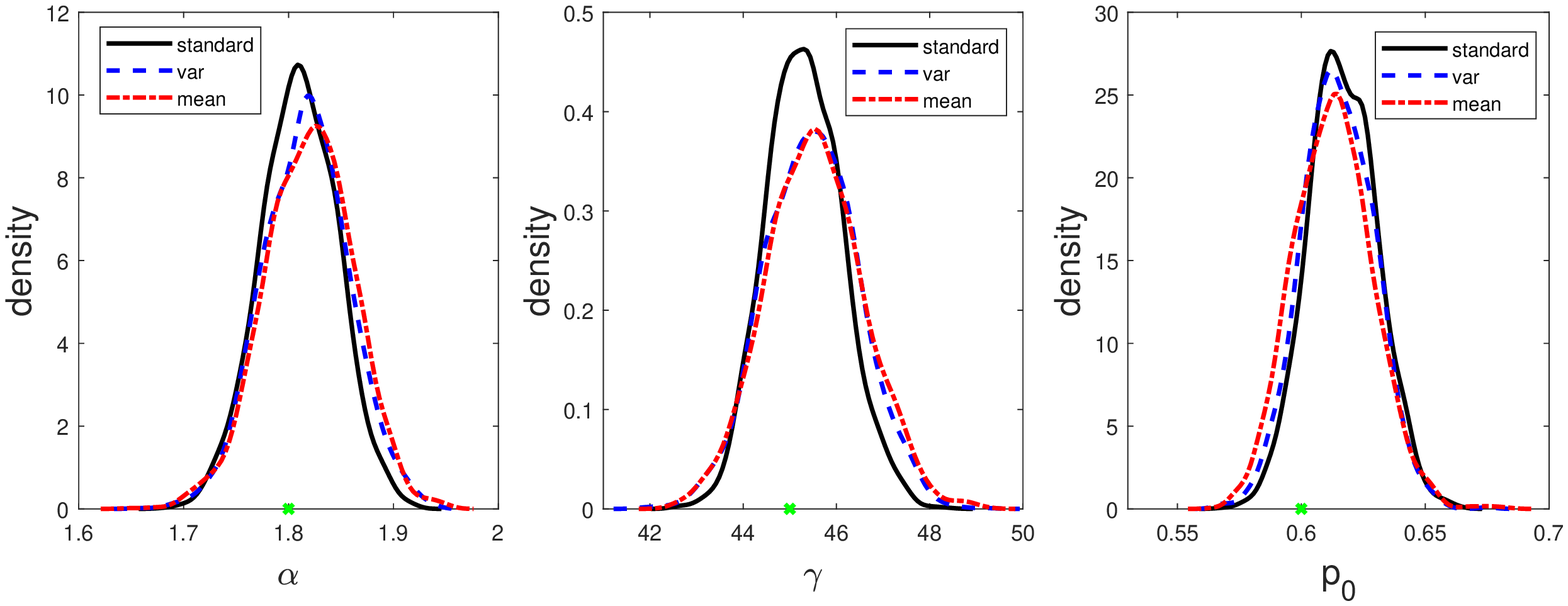} 
	 \caption{Univariate posterior distributions for the parameters when applying BSL (sold), R-BSL-V (dash) and R-BSL-M (dot-dash) to simulated data for the toad example.  True parameter values are shown as crosses.}%
	\label{fig:toad_model2_simdata_posteriors}
\end{figure}

\begin{figure}[h!]
	\centering
	\includegraphics[width=12cm, height=5cm]{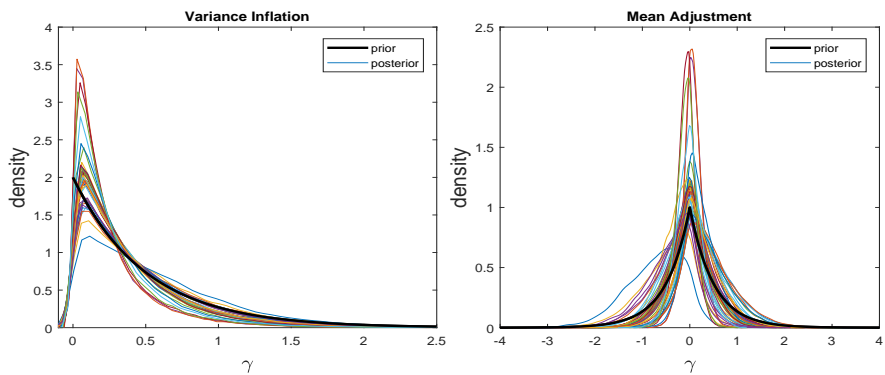} 
	 \caption{Posterior distributions for $\gamma$ for R-BSL-M (right) and R-BSL-V (left) applied to the simulated data for the toad example.  The thick line is the prior and the thin lines are the posteriors for the components of $\gamma$.}%
	\label{fig:toad_model2_simdata_gamma}
\end{figure}

For the real data, we required $n=2000$ simulations for estimating the synthetic likelihood to obtain an acceptance rate of 9\% for BSL.  However, the chain still suffered from periods of stickiness.  In contrast, with only $n=500$, the R-BSL-M and R-BSL-V produce acceptance rates of 7\% and 15\%, respectively, without substantial stickiness.  The variance inflation method offers a computational improvement of about one order of magnitude over BSL when accounting for both acceptance rate and number of simulations. 

The posterior distributions for the components of $\gamma$ for the R-BSL methods are shown in Figure \ref{fig:toad_realdata_gamma}.  It is evident from the plots that our methods have identified that there are three or four statistics that the model is not compatible with.  The statistic with the largest incompatibility is the number of returns for lag 1.  For R-BSL-V, the 95\% posterior predictive interval for this statistic is (262, 346) with an observed value of 234.  Other statistics showing some incompatibility are the first quantile differences of the non-returns for lags 3 and 4.  Figure \ref{fig:toad_model2_misspec_var} confirms that the observed data are not consistent with the posterior predictive distribution of the (log) non-return distances for lags 3 and 4, in that the model generally predicts larger non-return distances.  The mean adjustment results are similar (not shown).  Our adjustment methods permit in-depth analyses such as these and may provide practitioners valuable information for improving the model.  

\begin{figure}[h!]
	\centering
	\includegraphics[width=12cm, height=5cm]{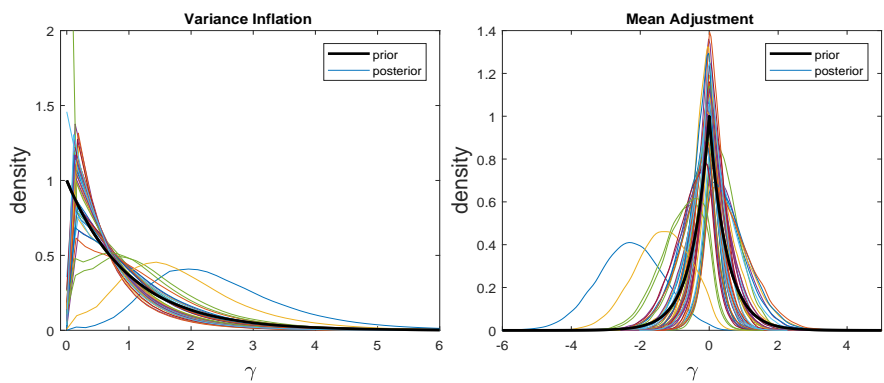}  
	\caption{Posterior distributions for $\gamma$ for R-BSL-M (right) and R-BSL-V (left)  applied to the real data for the toad example.  The thick line is the prior and the thin lines are the posteriors for the components of $\gamma$.}%
	\label{fig:toad_realdata_gamma}
\end{figure}

\begin{figure}[h!]
	\centering
	\includegraphics[width=12cm, height=6cm]{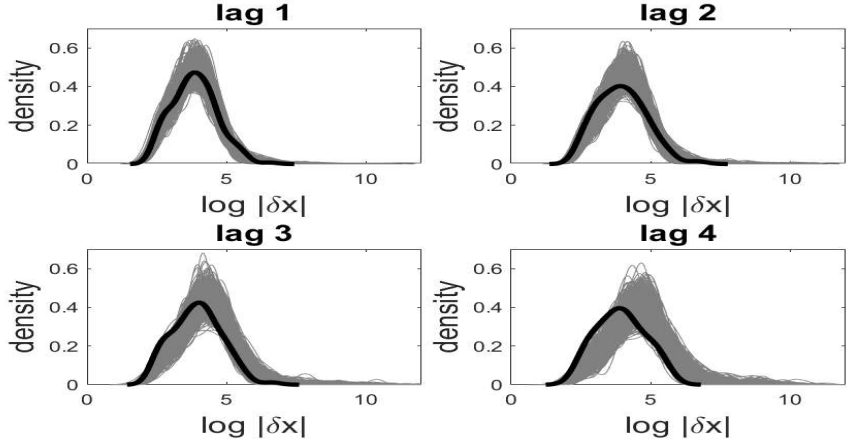} 
	 \caption{Posterior predictive distributions of the log non-returns for the four lags based on R-BSL-V.  The thick line is the distribution for the observed data.}%
	\label{fig:toad_model2_misspec_var}
\end{figure}

The bivariate posterior distributions for the parameters based on the adjustment methods are shown in Figure \ref{fig:toad_model2_mean_vs_var}.  It is evident that the estimated posterior distributions are similar, with the R-BSL-M posteriors slightly more concentrated than R-BSL-V.  Univariate posteriors for all BSL approaches are shown in Figure \ref{fig:toad_model2_posteriors}. {Comparing the univariate posteriors, we see that there is substantive disagreement between BSL and R-BSL. The R-BSL posteriors generally have fatter tails than the BSL posteriors, which indicates the presence of model misspecification, and are centered over different regions of the support (especially for $\alpha$ and $p_0$). This empirical evidence reinforces the analysis in Section \ref{sec:diffs}, where we argued that the R-BSL and BSL posteriors need not agree under model misspecification.}

\begin{figure}[h!]
	\centering
	\includegraphics[width=16cm, height=5cm]{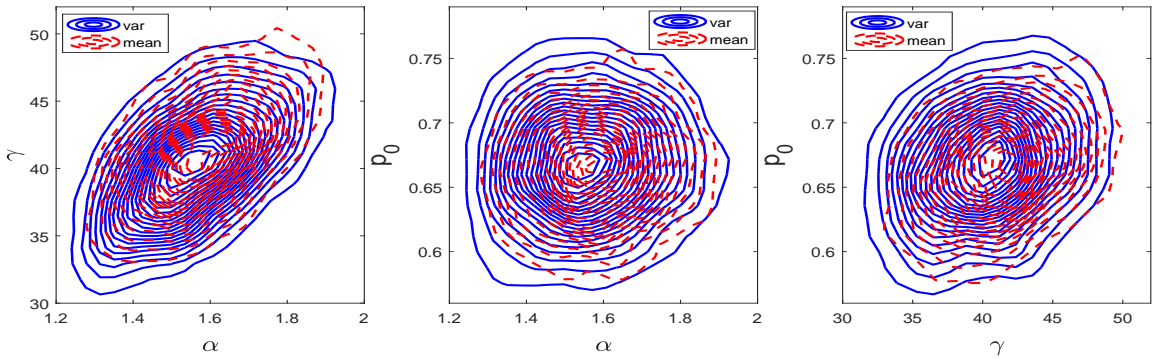}  
	\caption{Bivariate posterior distributions visualised as contour plots for the parameters based on R-BSL-M (right) and R-BSL-V (left) applied to the real data for the toad example.  }%
	\label{fig:toad_model2_mean_vs_var}
\end{figure}

\begin{figure}[h!]
	\centering
	\includegraphics[width=16cm, height=5cm]{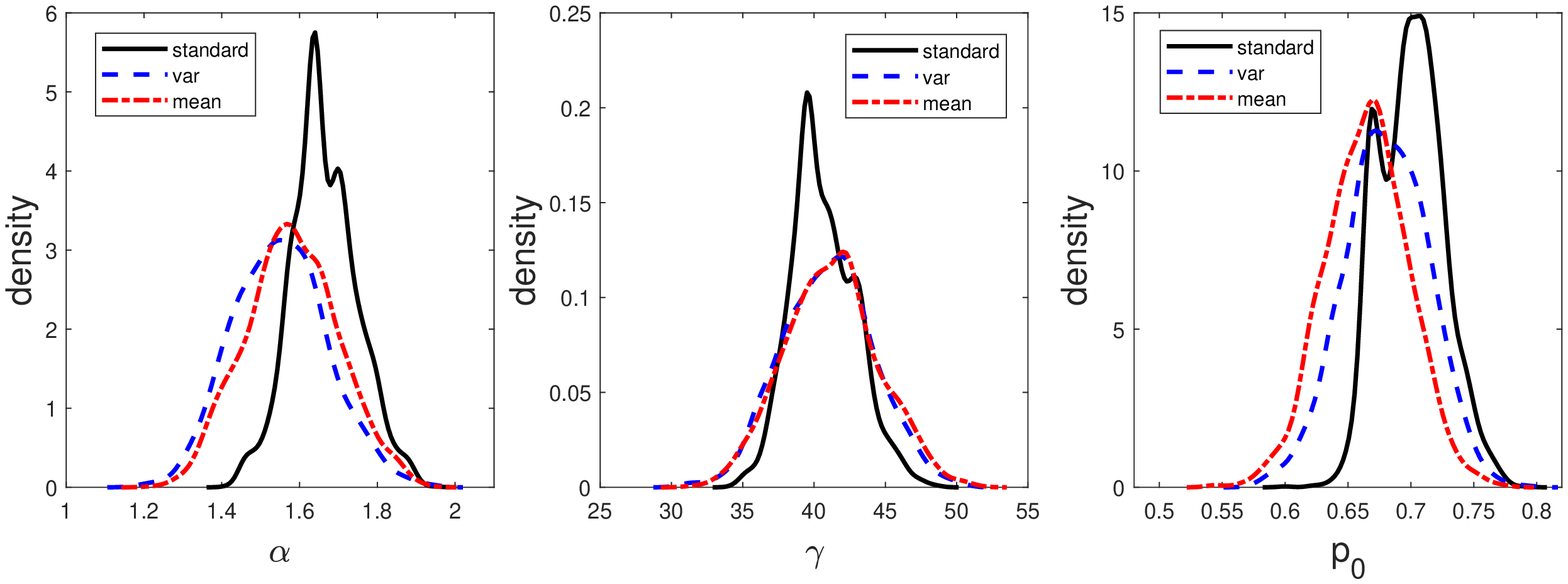}
  \caption{Univariate posterior distributions for the parameters when applying BSL (sold), R-BSL-V (dash) and R-BSL-M (dot-dash) to real data for the toad example.}%
	\label{fig:toad_model2_posteriors}
\end{figure}

\section{Discussion} \label{sec:discussion}
This paper has made two significant contributions to the literature on approximate Bayesian methods. Firstly, to our knowledge, this is the first piece of research to demonstrate that, similar to approximate Bayesian computation (ABC), Bayesian synthetic likelihood (BSL) can deliver unreliable inference when the assumed model is misspecified. Secondly, to circumvent the poor behavior of BSL in these settings, we have proposed a modification of BSL that displays robustness to model misspecification. Several Monte Carlo and empirical examples are used to illustrate the performance of this new method, with the results demonstrating both the statistical and computational benefits of this new approach when the model is misspecified. 

In addition to delivering more accurate statistical inference under model misspecification, this new approach also allows the user to detect precisely which summary statistics are incompatible with the assumed data generating process. Incorporating this information within subsequent rounds of model building could lead to better models that can more accurately capture the behavior exhibited by the observed summary statistics. In this sense, the robust BSL approach can be viewed as a BSL version of the model criticism approach of \cite{ratmann2009model}. In the context of ABC, \cite{ratmann2009model} propose an approach to detect aspects of the model that the summary statistics can not adequately capture. Their approach relies on treating the ABC tolerance as an unknown parameter, and augmenting the original ABC inference problem with this additional parameter. The authors argue that posterior realizations for the tolerance parameter that are ``large''  indicate the possibility of a mismatch between the model and the observed data. 

While useful, in the case of multivariate summaries the approach of \cite{ratmann2009model} requires a tolerance parameter for each summary statistic used in the analysis, with posterior inference then required on the full set of model parameters and tolerance parameters. {Therefore, even for a moderate number of summaries, this approach can exacerbate the underlying curse-of-dimensionality in ABC, as it pertains to {both the dimension of summaries and the number of parameters in the analysis}. For instance, the empirical example in Section \ref{sec:toad} employed 48 summary statistics, which is much larger than can reliably be considered by standard implementations of ABC.} In addition, the approach considered herein has as a direct benchmark with which to gauge the impact of misspecification on the summaries, namely the prior distribution of the adjustment components. If the corresponding posterior for the adjustment component in the robust version of BSL does not resemble the prior, this is strong evidence that this summary can not be matched by the assumed model. If one wished to put a numerical value, or conduct a formal hypothesis test, on the difference between the prior and posterior, any number of techniques could be used. 

The examples illustrate that, in particular, the variance inflation approach can significantly improve the MCMC acceptance rate, under model misspecification (relative to standard BSL).  The variance inflation approach bears some resemblance to MCMC ABC approaches that assign a distribution to the ABC tolerance to facilitate MCMC mixing by proposing a relatively large tolerance value (e.g., \citealt{Bortot2007}).  However, improving mixing is not our primary focus, but is simply a useful by-product.  Our R-BSL approaches may also be useful for initial explorations of the parameter space when it is not known where the bulk of the posterior support is since, even for a correctly specified model, a poor parameter value will not be able to recover the observed statistic.

As with BSL, our R-BLS approach requires that a Gaussian approximation to the summaries is reasonable. In cases where it is not, it is possible that our methods may detect incompatibility when it is not present.  For example, the variance inflation parameter may activate to accommodate a summary statistic distribution with a heavy tail, even when the model is correct. We are currently working on adapting our approach to other likelihood-free methods that relax the Gaussian assumption, such as ABC and the semi-parametric extension of BSL \citet{an2018robust}, which uses flexible models for the marginal summary statistic distributions.     

{The choice of prior for $\gamma$ was chosen for simplicity and ease of posterior sampling via the slice sampler.    We note, however, that other prior choices are possible.  For example, it is worth investigating sparsity inducing priors such as spike and slab priors.   However, the mixed discrete-continuous nature of these priors may complicate posterior sampling in our context.  Another possible choice is the class of global-local shrinkage priors such as the horseshoe.     There is also a question of hyperparameter choice, although our default choice is sensible and produced good empirical results in our examples.  Attempting to tune the hyperparameter by cross validation or running prior sensitivity studies would be highly computationally intensive in our likelihood-free context.  We leave a more thorough investigation of prior choice to future research.}

Lastly, we note that, given the similarities between BSL and ABC, a natural question posed during this research was  whether or not the mean and variance adjustment approaches discussed in this current paper were applicable in the context of ABC. In concurrent work from the  authors, preliminary investigations into a similar type of mean and variance adjusted ABC have revealed that such an approach can mitigate the poor performance of ABC under model misspecification (\citealp{frazier2020model}). 

\section*{Acknowledgments}

The authors are grateful to Ziwen An, who provided some useful code for this paper. Frazier was supported by the Discovery Early Career Researcher Award funding scheme.

\singlespacing
\bibliographystyle{apalike} 
\bibliography{refs}

\appendix 
\section*{Appendix}

This appendix contains additional details for the contaminated normal example and an empirical example that analyzes the behavior of a popular collective cell spreading model. In addition, this material includes the proof of Proposition 1 in the main text. 

\section{Additional Details: Normal Example}
In this section, we present additional details for Example 1 in the paper. First, we analyze the R-BSL adjustment components, $\Gamma$, from Example 1, and then we present repeated sampling results for BSL and R-BSL across three levels of model misspecification. We refer the reader to Sections \ref{sec:toynorm} and \ref{sec:toynormcont} of the main paper for details of the Monte Carlo specification.

\subsection{Adjustment Components}
Figure \ref{fig:normal_rbsl_gamma} displays the resulting posterior densities for $\Gamma$ across the two R-BSL procedures, and across all levels of misspecification. Panels A and B give the results for R-BSL-M, and correspond to the components $\gamma_1$ and $\gamma_2$, respectively, while panels C and D give the same results for R-BSL-V. For comparison purposes, the black line in each panel represents the prior densities, and the color-coding in each figure represents the level of misspecification (where $\sigma^2_\epsilon=1$ encodes correct model specification). 

Focusing on Panel A, we see that the posterior densities for the $\gamma_1$ component, in R-BSL-M, which captures our ability to match the first observed statistic (the mean), are indistinguishable from the prior across all the replicated data sets, which implies that we can match the mean of this model regardless of model misspecification. In contrast, in Panel B we see that the second component, which captures our ability to match the second observed statistic (the variance), looks nothing like the prior, except perhaps at low levels of misspecification.

Panels C and D describe precisely the same story as in Panels A and B but correspond to $\gamma_1$ and $\gamma_2$ in R-BSL-V. Under correct specification ($\sigma^2_\epsilon=1$) the posteriors are indistinguishable from the priors, we are easily able to detect departures from compatibility for the second summary statistic, and the posteriors for the first adjustment term remain indistinguishable from the prior across all data sets. These results demonstrate that both R-BSL approaches are capable of reliably detecting which features of the model we are not able to replicate.

\begin{figure}[h!]
	\centering
	\includegraphics[width=17cm, height=9cm]{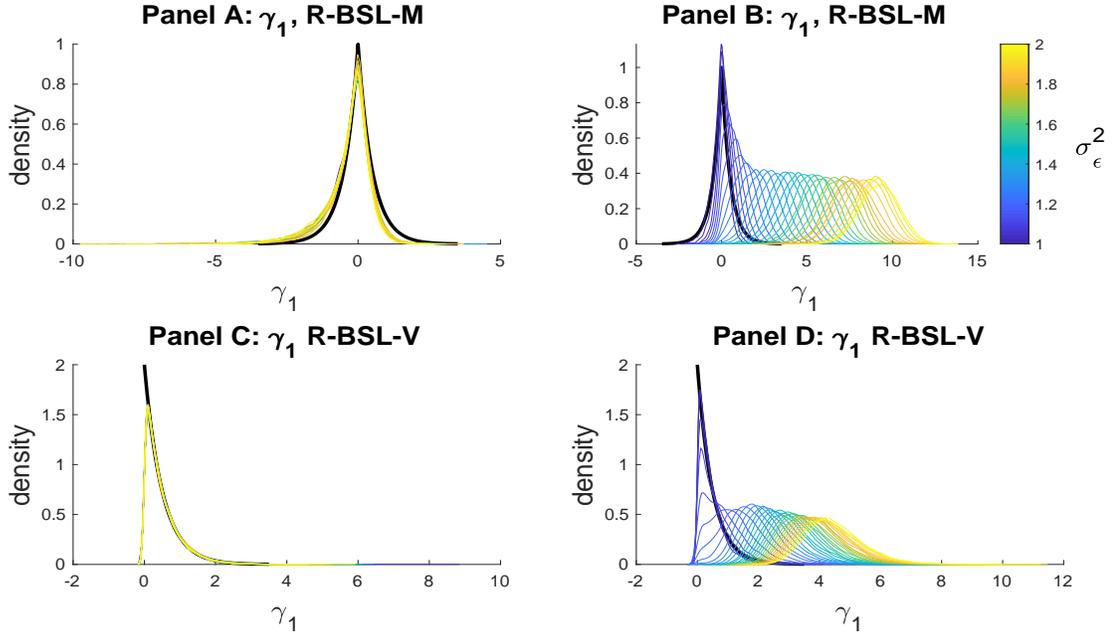}  
	\caption{\footnotesize{Marginal posteriors for adjustment components. Panels A and B correspond to the mean adjustment R-BSL approach (R-BSL-M), while Panels C and D correspond to the variance adjustment R-BSL approach (R-BSL-V). Panels A and C correspond to the components for the first summary, and panels B and D for the second summary.} In each sub-figure, the black line corresponds to the prior density for each component.}%
	\label{fig:normal_rbsl_gamma}
\end{figure}

\subsection{BSL and R-BSL Comparison}

We now analyze the repeated sampling behavior of R-BSL and compare it to BSL. We again consider data generated from the contaminated normal model, and we choose $\sigma^2_\epsilon$ in equation \eqref{eq:truemod} in the main text so that the observed data for $\y$ has \textit{sample variance} equal to $\{1.00,1.50,2.00\}$. For each value of the sample variance, we generate one hundred replicated data sets. We then apply R-BSL and BSL, and calculate the bias (BIAS), and root mean squared error (RMSE), associated with the posterior means, as well as the average credible set length (LEN), obtained from a 95\% confidence set, and the Monte Carlo coverage (COV), as determined by a 95\% confidence set. Table \ref{tab:bsl1} displays the results and demonstrates that R-BSL performs well across both correctly and incorrectly specified models. In the case of correct specification, R-BSL-V and BSL perform very similarly. However, as model misspecification increases, R-BSL-V yields more accurate point estimators than BSL, and better uncertainty quantification. These results demonstrate that when the model is misspecified, R-BSL will yield more reliable statistical inferences than those obtained by BSL. 
\begin{table}[h!]
	\centering
	\caption{Summary measures for posterior mean accuracy, calculated as averages across the Monte Carlo replications. RMSE- root mean squared error of the posterior mean, BIAS- bias of the posterior mean, LEN- credible set length, COV- Monte Carlo coverage. BSL refers to BSL, R-V refers to R-BSL-V, and R-M refers to R-BSL-M. Average acceptance rates across the three different levels of misspecification are as follows: R-BSL-V:  71.86\%, 53.47\%, 41.78\%; R-BSL-M: 68.44\%, 24.51\%, 5.83\%; and BSL: 68.77\%, 3.93\%, 0.02\%}
	\begin{tabular}{lrrrrrrrrrrr}
		&       & \multicolumn{1}{l}{$\sigma$=1.0} &       &       &       & \multicolumn{1}{l}{$\sigma$=1.5} &       &       &       & \multicolumn{1}{l}{$\sigma$=2.0} &  \\
		& \multicolumn{1}{l}{BSL} & \multicolumn{1}{l}{R-V} & \multicolumn{1}{l}{R-M} &       & \multicolumn{1}{l}{BSL} & \multicolumn{1}{l}{R-V} & \multicolumn{1}{l}{R-M} &       & \multicolumn{1}{l}{BSL} & \multicolumn{1}{l}{R-V} & \multicolumn{1}{l}{R-M} \\
		RMSE  & 0.0022 & 0.0034 & 0.0292 &       & 0.0175 & 0.0029 & 0.0299 &       & 0.1436 & 0.0026 & 0.0257 \\
		BIAS  & 0.0000 & 0.0003 & 0.0296 &       & -0.0006 & 0.0006 & 0.0297 &       & 0.0165 & 0.0003 & 0.0286 \\
		LEN   & 0.3917 & 0.4807 & 0.5065 &       & 0.4515 & 0.4821 & 0.5120 &       & 0.1916 & 0.4826 & 0.5185 \\
		COV   & 100\% & 100\% & 100\% &       & 100\% & 100\% & 100\% &       & 52\% & 100\% & 100\% \\
	\end{tabular}
	\label{tab:bsl1}
\end{table}%

\section{Collective Cell Spreading }
\subsection{Background}

Collective cell spreading models are often used to gain insight into the biological mechanisms governing, for example, wound healing and skin cancer growth (e.g.\ \cite{VoDiameter2014,Vo2015}).   \citet{Browning2018} develop a simulation-based model where cells are able to move freely in continuous space.  They calibrate the model to real \emph{in vitro} data collected from a cell proliferation assay experiment using a rejection-based ABC algorithm.  Here, using our new synthetic likelihood methods, we demonstrate that the model is not compatible with the observed summary statistic and provide insight into what aspects of the data that the model is not able to recover.

The model of  \citet{Browning2018} is a stochastic individual based model where cells move and interact in a two dimensional space.  Here we provide only brief details of the model and refer to \citet{Browning2018} for the full description.  Proliferation (cell birth) and motility (movement) for each cell evolves in continuous time according to a Poisson process.  The intrinsic rates are given by $p$ and $m$ for proliferation and motility events, respectively.  The rates of these processes are also neighbourhood-dependent, with rates decreasing as the amount of crowding around a cell increases.  The closeness of cells is governed by a Gaussian kernel that depends on a fixed cell diameter, $\sigma$.  When a cell proliferates, it places a new cell randomly in its neighbourhood according to an uncorrelated two dimensional Gaussian centered at the cell location with component variances of $\sigma^2$.  When a motility events occurs, the cell moves a distance of $\sigma$. The direction of the move depends on cell density, biased towards lower cell density.  A parameter used to help determine the move direction, $\gamma_b$, is part of a Gaussian kernel used to measure the closeness of cells.  The parameter of interest is $\theta = (p,m,\gamma_b)^{\top}$.

In the experiments of \citet{Browning2018}, images of the cell population are taken every 12 hours starting at 0 hours with the final image taken at 36 hours.  \citet{Browning2018} use the number of cells and the pair correlation computed from each image as the summary statistics, resulting in a six dimensional summary statistic here.  The pair correlation is the ratio of the number of pairs of agents separated by some pre-specified distance to an expected number of cells separated by the same distance if the cells were uniformly distributed in space.   

The prior distribution is set as $p \sim \mathcal{U}(0,10)$, $m \sim \mathcal{U}(0,0.2)$ and $\gamma_b \sim \mathcal{U}(0,20)$ with no  dependence amongst parameters, as in \citet{Browning2018}.  We use MCMC to sample the posterior with 50,000 iterations and no burn-in as we initialise the chain at the point estimate $\theta = (1,0.04,6)^\top$ reported in \citet{Browning2018}.  We sample over the space of a logit-type transformation of $\theta$ so that any proposal is within the prior bounds.   We use a multivariate Gaussian random walk proposal on the transformed space with a covariance matrix obtained via some pilot MCMC runs.  We use $m=50$ model simulations to estimate the synthetic likelihood at each MCMC iteration.  We run our methods on both simulated (using the point estimate of \citet{Browning2018}) and real data.

\subsection{Results}

For the priors for each component of $\gamma$, we use a Laplace distribution with a scale of 0.5 for the mean adjustment method and an exponential prior with mean of 0.5 for the variance inflation method.  

Firstly we present results for the simulated data (model correctly specified).  As shown in Figures \ref{fig:lattice_contour_simulated_mean} and \ref{fig:lattice_contour_simulated}, the posterior distributions on $\theta$ are similar regardless of whether BSL or R-BSL is applied.  The MCMC acceptance rates for BSL and R-BSL-M are both 21\% and 20\%, respectively.   The variance inflation seems to allow for a slightly increased acceptance rate (24\%) compared to mean adjustment.

\begin{figure}[h!]
	\centering
	\includegraphics[width=16cm, height=6cm]{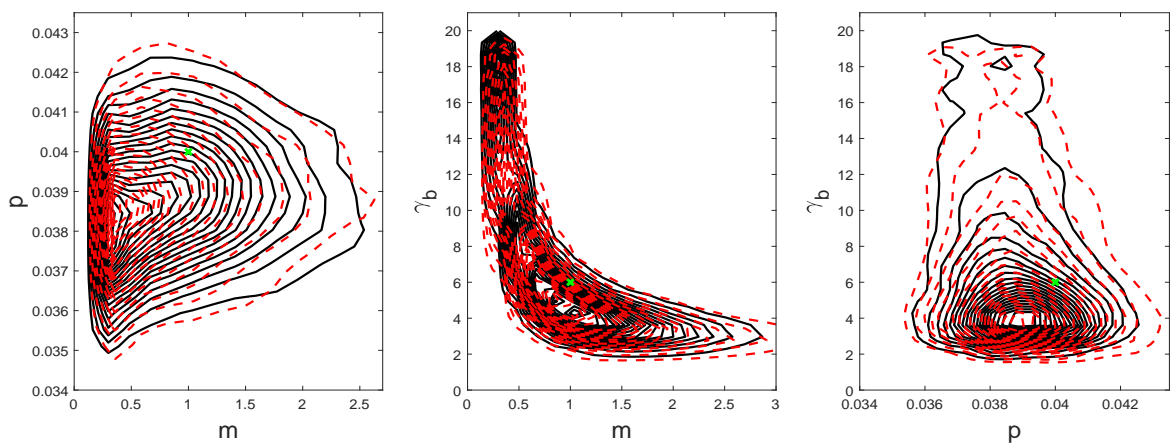}  
	\caption{Contour plots of the posterior distributions based on the simulated data for the collective cell spreading example.  Results are shown for BSL (solid) and R-BSL-M (dash).  The true parameter values are shown as a crosses.}%
	\label{fig:lattice_contour_simulated_mean}
\end{figure}

\begin{figure}[h!]
	\centering
	\includegraphics[width=16cm, height=6cm]{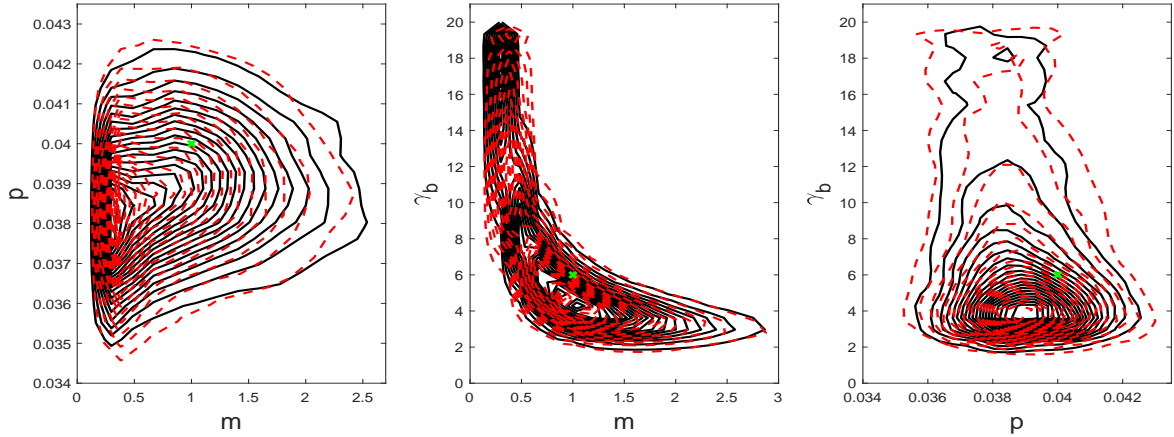}  
	\caption{Contour plots of the posterior distributions based on the simulated data for the collective cell spreading example.  Results are shown for BSL (solid) and R-BSL-V (dash).  The true parameter values are shown as crosses.}%
	\label{fig:lattice_contour_simulated}
\end{figure}

The posterior distribution for each component of $\gamma$ is shown in Figure \ref{fig:lattice_gamma_simulated_mean}.   It can be seen that most posteriors are similar to the prior.  For both R-BSL methods, there is no indication that any of the statistics are incompatible with the model, as expected.

\begin{figure}[h!]
	\centering
	\includegraphics[width=12cm, height=6cm]{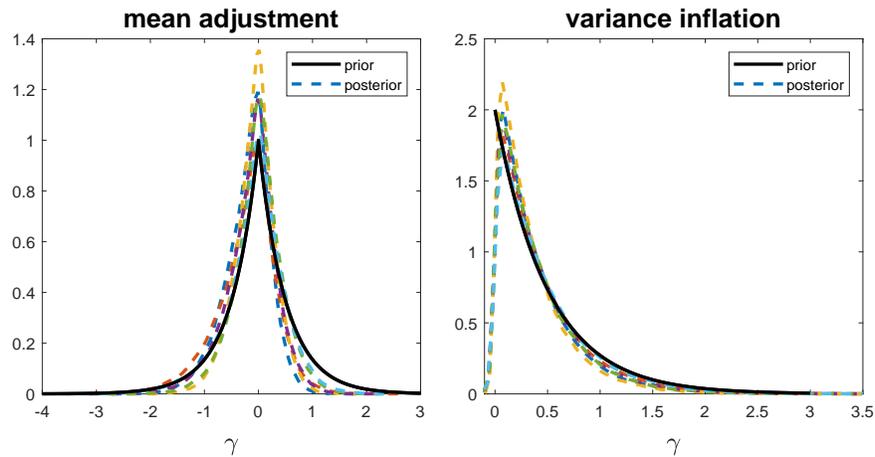}  
	\caption{Posterior distributions for each component of $\gamma$ (dashed lines) based on the simulated data for the collective cell spreading example for R-BSL-M (left) and R-BSL-V (right).  The prior distribution of $\gamma$ is shown as solid lines for both R-BSL-M (left) and R-BSL-V (right).}%
	\label{fig:lattice_gamma_simulated_mean}
\end{figure}

For the real data, the MCMC acceptance rate using BSL is only 3\% as the variance of the synthetic likelihood is high generating long periods of no acceptance.  Applying R-BSL-M and R-BSL-V results in an MCMC acceptance rate of roughly 12\% and 18\%, respectively, permitting statistical inference.  Again, the variance inflation seems to produce an improved acceptance rate.

The univariate posterior distributions for $\gamma$ for R-BSL-M and R-BSL-V are shown in Figures \ref{fig:lattice_gamma_posteriors_real} and \ref{fig:lattice_epsilon_posteriors_real}, respectively.  Both methods identify that the model is not compatible with the pair correlation statistic at 12 and 36 hours.

\begin{figure}[h!]
	\centering
	\includegraphics[width=12cm, height=7cm]{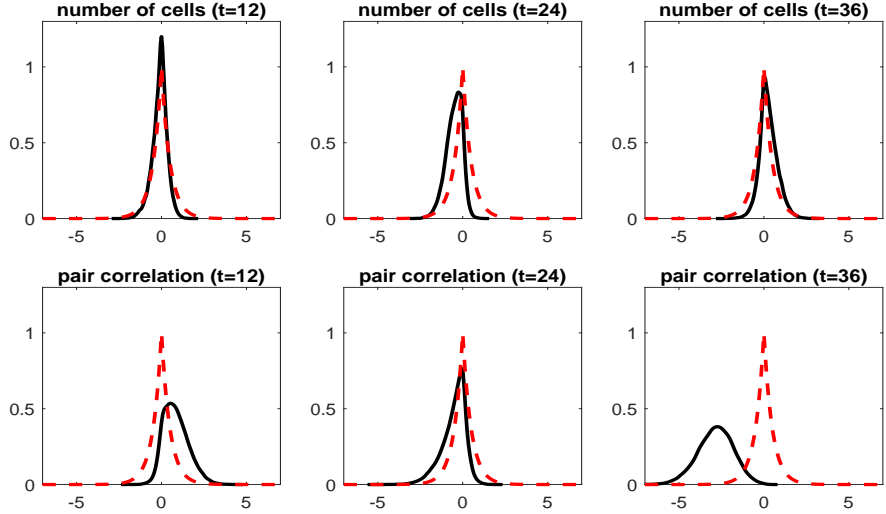}  
	\caption{Posterior distributions (solid) for each component of $\gamma$ when applying R-BSL-M to the real data of the collective cell spreading example.  The prior distributions, which are Laplace distributed with scale 0.5, are also shown (dash).}%
	\label{fig:lattice_gamma_posteriors_real}
\end{figure}

\begin{figure}[h!]
	\centering
	\includegraphics[width=12cm, height=7cm]{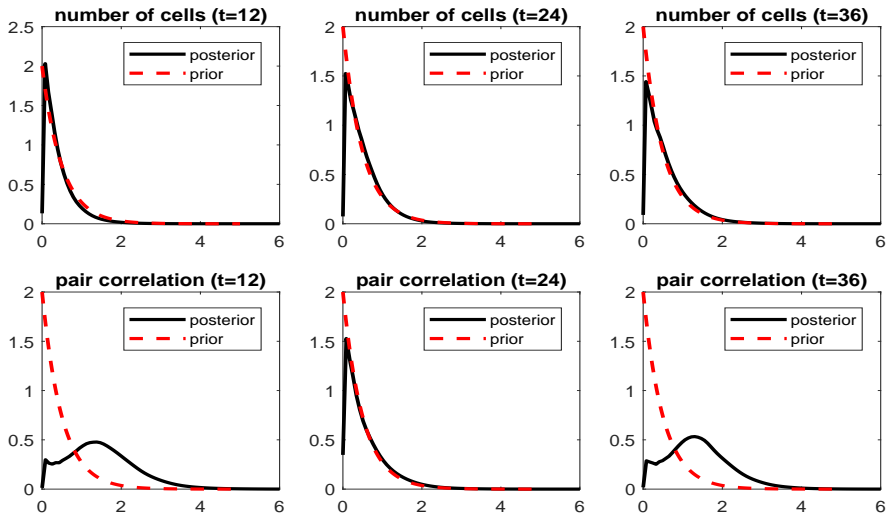}  
	\caption{Posterior distributions (solid) for each component of $\gamma$ when applying R-BSL-V to the real data of the collective cell spreading example.  The prior distributions, which are exponential with mean 0.5, are also shown (dash).}%
	\label{fig:lattice_epsilon_posteriors_real}
\end{figure}

In Figure \ref{fig:lattice_postpredictions_real05_mean}, we show R-BSL-M posterior predictive distributions of the summary statistics without (left) and with (right) using estimated mean adjustment parameters with the observed summaries overlaid.  The corresponding plots for the variance inflation is shown in Figure \ref{fig:lattice_postpredictions_real}.  From both figures, it is evident from the plots on the left that the model is successful in tracking the number of cells over time.  However, the model underestimates the rate of decrease in the pair correlation over time.  This is valuable information that might enable mathematical biologists to extend the model so that this data feature can be better captured.  Both adjustment methods have allowed us to make this inference.    It can be seen in the second column of Figure \ref{fig:lattice_postpredictions_real05_mean} that the mean adjustment is able to shift the predictions so that the observed statistic does not lie so far in the tails.  From the second column of Figure \ref{fig:lattice_postpredictions_real},  the variance adjustment expands the predictions so that the observed statistic does not lie so far in the tails.

\begin{figure}[h!]
	\centering
	\includegraphics[width=16cm, height=10cm]{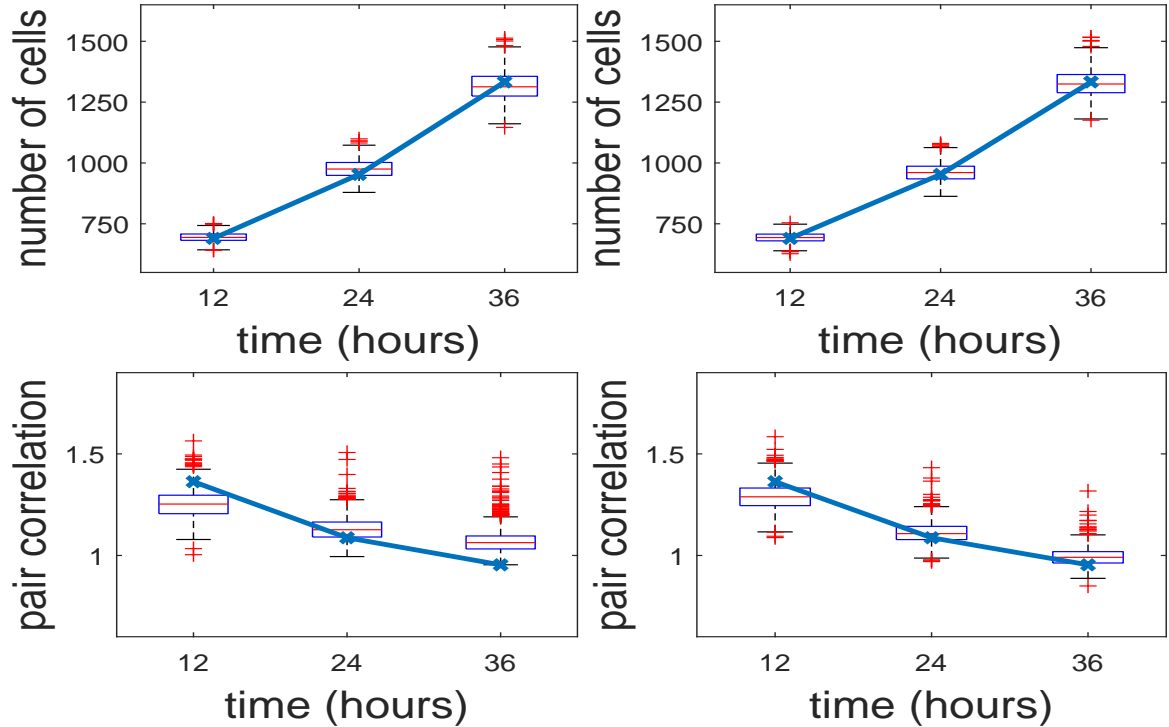}  
	\caption{Posterior predictive distributions (summarised as boxplots) of the summary statistics obtained with R-BSL-M.  Shown are the posterior predictive distributions of the summaries without (left) and with (right) including the variance inflation in the predictions. The observed summary statistics are overlaid as crosses.}%
	\label{fig:lattice_postpredictions_real05_mean}
\end{figure} 

\begin{figure}[h!]
	\centering
	\includegraphics[width=16cm, height=10cm]{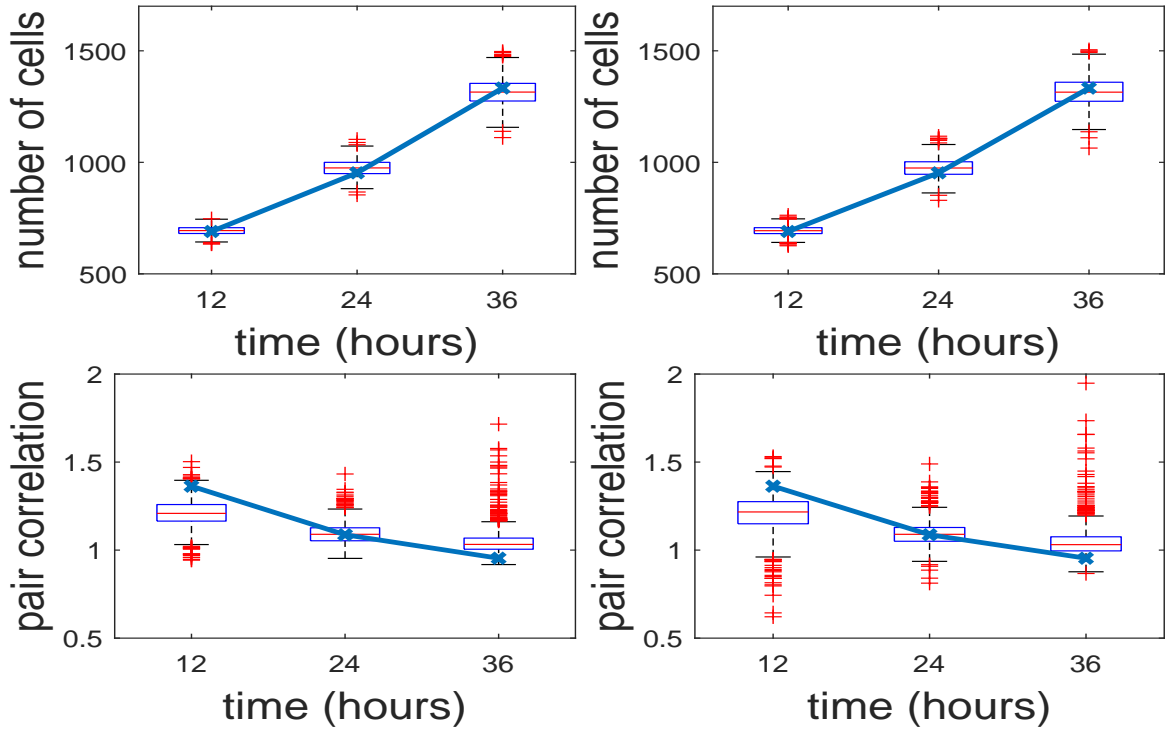} 
	\caption{Posterior predictive distributions (summarised as boxplots) of the summary statistics obtained with R-BSL-V.  Shown are the posterior predictive distributions of the summaries without (left) and with (right) including the variance inflation in the predictions. The observed summary statistics are overlaid as crosses.}%
	\label{fig:lattice_postpredictions_real}
\end{figure}

Finally, Figure \ref{fig:lattice_posteriors_real_mean_vs_variance} compares the posterior distributions for the mean adjustment and variance inflation, together the standard synthetic likelihood results.  It can be seen that the posterior distributions are broadly similar, except that the standard synthetic likelihood results suffer from substantial Monte Carlo error due to the small acceptance rate.

\begin{figure}[h!]
	\centering
	\includegraphics[width=16cm, height=5cm]{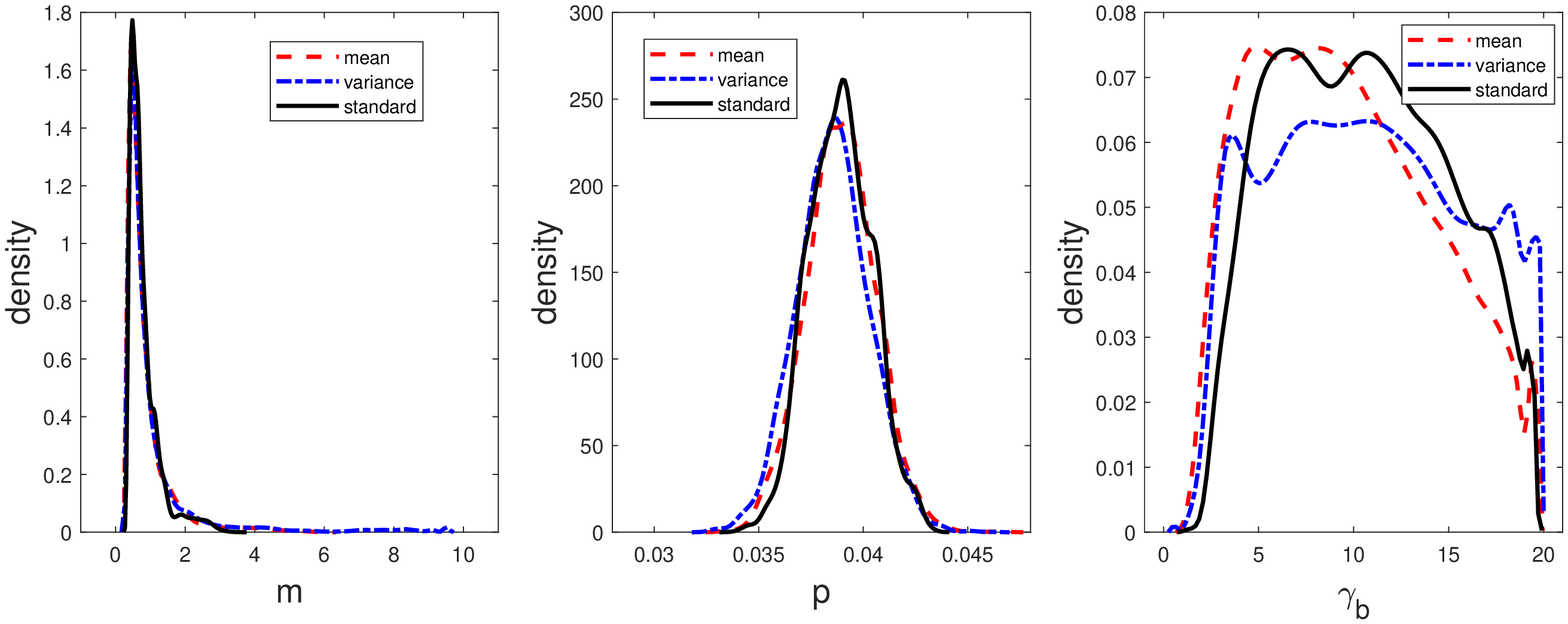}  
	\caption{Posterior distributions for each component of $\theta$ when applying BSL (solid), R-BSL-M (dash) and R-BSL-V (dot-dash) to the real data of the collective cell spreading example.  The priors used are Laplace with scale 0.5 for mean adjustment and exponential with a mean of 0.5 for variance inflation.}%
	\label{fig:lattice_posteriors_real_mean_vs_variance}
\end{figure}

\section{Theoretical Properties of R-BSL}
For clarity, we first recall the assumptions and the result in question. 

\begin{assumption}\label{ass:one}{
		There exists a sequence of positive real numbers $v_n$ diverging to $\infty$ such that, for some distribution $Q$ on $\mathbb{R}^{d_\eta}$ and some vector $b_0\in\mathbb{R}^{d_\eta}$, $$v_n\left[\eta(\y)-b_0\right]\Rightarrow Q,\text{ \em{under} } P^0_n.$$}
\end{assumption}
\begin{assumption}\label{ass:two}
	(i) The sequence $\{v_n\}_{n\ge1}$ is such that, for all $\theta\in\Theta$ and some $n$ large enough, there exists constants ${c}_1,{c}_2$, $c_1\leq c_2$, such that, for $\|\cdot\|_*$ denoting a matrix norm, $0<{c}_1\leq \|v_{n}^{}{A}^{}_{n}({\theta }%
	)\|_{*}\leq {c}_2<\infty$; (ii) For all $\theta\in\Theta$ and all $n\ge1$, the $d\times d$-matrix $A_{n}(\theta)$ is continuous in $\theta$.
\end{assumption}
\begin{assumption}\label{ass:three}
	There exists a deterministic map ${b}:{\Theta }%
	\rightarrow \mathcal{B}$, such that, for all $\theta\in\Theta$, and for constants $\alpha, u_0>0$, for all $0<u<u_0 v_n$,
	\begin{equation*}
	{G}_n\left[ \|v_n\{\eta(\y)-b(\theta)\}\|>u\mid\theta\right] \leq c(\theta) u^{-\alpha},
	\end{equation*}uniformly for $n\ge1$ and where $\int_\Theta c(\theta)\pi(\theta)\text{d}\theta=O(1). $
\end{assumption}

\begin{assumption}\label{ass:four}(i) There exists some $\tau>0$ such that, for all $0<u <u_0 v_n$, the prior probability satisfies
	$$\Pi \left[ \|{b}({\theta })-{b}_{0}\|\leq u \right]
	\asymp u ^{\tau}.$$(ii) The prior density $\pi(\theta)$ is continuous and satisfies $\pi(\theta_0)>0$. 
\end{assumption}

\begin{assumption}\label{ass:five} {(i)} The map $\theta\mapsto{b}(\theta)$
	is continuous and injective, with $b(\theta_0)=b_0$ for some $\theta_0\in\Theta$, and satisfies:
	$\Vert {\theta }-{\theta }_{0}\Vert \leq L\Vert {b}({\theta })-b_0\Vert ^{\kappa }$ on some open neighbourhood of ${\theta }_{0}$ with $L>0$ and $\kappa >0$.
\end{assumption}

\begin{assumption}\label{ass:six} If Assumption \ref{ass:five} is satisfied, for any $\epsilon>0$, there exists $u,\delta>0$ and a set $V_n$ such that for all $\theta\in\{\theta:\|b(\theta)-b_0\|\leq u v_n^{-1}\}$ $$V_n\subset \left\{\eta\in\mathbb{R}^{d_\eta}:g_n^0(\eta)\lesssim {g}_n\left(\eta\mid \theta\right)\right\}\text{ where } P_n^0(V_n^c)<\epsilon.$$
	
\end{assumption}

Generally speaking, Assumptions \ref{ass:one}-\ref{ass:six} impart regularity on the summary statistics needed to deduce a posterior concentration result. Assumptions \ref{ass:one}, \ref{ass:three} and \ref{ass:six} are similar to those used by \cite{marin2014relevant} to deduce concentration of posteriors conditioned on summary statistics, while Assumptions \ref{ass:two}, and \ref{ass:four} are specific to the analysis of BSL, and have also been used in \cite{frazier2019bayesian}. We now discussion each of the assumptions in detail.

Assumption \ref{ass:one} imposes regularity on the observed summary statistics and requires that they satisfy a converge in distribution result at rate $v_n$, but does not restrict this distribution to be Gaussian.  Assumption \ref{ass:two} is specific to BSL and requires that the BSL variance matrix be well-behaved for all values of $\theta\in\Theta$. This assumption is needed to rule out cases where the BSL variance does not exist or is ill-conditioned (which would be the case, e.g.,  if the summary statistics are perfectly correlated). Assumption \ref{ass:three} is a condition on the tails of the simulated summary statistics and requires that they have at least a polynomial tail (i.e., it requires the existence of at least $\alpha$ moments) uniformly in $\theta$.\footnote{A stronger version of this condition has been used in \cite{frazier2019bayesian} to deduce posterior concentration of the BSL posterior.} Such a condition allows the application of Markov-type inequalities, which are a key ingredient in many posterior concentration results. Assumption \ref{ass:four} is a condition on the tails of the prior used in BSL, and requires that the tails of the prior are not too thick. This condition is satisfied by any prior that admits an exponential moment. Assumption \ref{ass:five} is an identification condition and requires that the simulated summaries are capable of replicating the asymptotic mean of the observed summary statistics. Together with Assumptions \ref{ass:one} and \ref{ass:four}, Assumption \ref{ass:five} ensures that the compatibility condition is satisfied. Assumption \ref{ass:five}, or a similar variant, has been used in several studies on the asymptotic behavior of approximate Bayesian procedures (see \citealp{fearnhead2018asymptotics} for a detailed discussion on this condition). Assumption \ref{ass:six} requires that, for values of $\theta$ such that $\|b(\theta)-b_0\|$ is small, up to a universal constant, the assumed model density can be bounded below by the true model density, pointwise, and that the support over which this bound is satisfied has large probability. This condition allows us to link the assumed and true model within the theoretical analysis. This condition is vacuously satisfied if the assumed and true model coincide. However, such a condition would be overly restrictive since compatibility does not require that the assumed model, $g_n(\cdot|\theta)$,  coincides with the true model, $g_n^0(\cdot)$, but only that certain moments of the two models agree.

We note that Assumptions \ref{ass:one}, \ref{ass:three} and \ref{ass:six} are similar to Assumptions 1, 2 and 4 imposed by \cite{marin2014relevant} in their analysis of posteriors conditioned on summary statistics, while variants of Assumptions \ref{ass:one}-\ref{ass:five} have been used by \cite{frazier2019bayesian} in the analysis of the asymptotic properties of BSL. Furthermore, we note that it is trivial to verify Assumptions \ref{ass:one}-\ref{ass:six} for the contaminated normal example when compatibility is in evidence (i.e., when $\sigma^2_\epsilon=1$). In addition, we recall that in the moving average example the model is not compatible, and we have already verified that Assumptions \ref{ass:five} is not satisfied. That being said, we note that by following the analysis in Example 1 of \cite{frazier2016asymptotic}, Assumptions \ref{ass:one}-\ref{ass:four} can be verified for this example. 

The following result gives the theoretical behavior of the R-BSL posterior under the above assumptions. The proof of the result follows. 

\begin{proposition}\label{prop:one} 
	Under Assumption \ref{ass:one}-\ref{ass:six}, for any $\delta>0$, 
	$$
	\Pi\left[\|b(\theta)-b_0\|\leq\delta |\eta(\y)\right]=1+o_P(1).
	$$
	Moreover, for any $A\subseteq\mathcal{G}$: $$\Pi\left[\Gamma\in A|\eta(\y)\right]=\Pi[\Gamma\in A]+o_P(1).$$
\end{proposition}

We first prove a result that is of independent interest. Namely, we demonstrate that under Assumptions \ref{ass:one}-\ref{ass:six}, the standard BSL posterior concentrates all posterior mass onto the sets of the form $\{b:\|b-b_0\|\lesssim v^{-1}_n\}$.  To simplify the computations, we demonstrate this result for the so-called ``idealized'' BSL posterior, which takes as the mean and variance the infeasible counterparts 
$$
b(\theta):=\mathbb{E}[\eta(\y)|\theta]\text{ and }A_{n}(\theta):=\left(\mathbb{E}\left[\left(\eta(\y)-\mathbb{E}[\eta(\y)|\theta]\right)\left(\eta(\y)-\mathbb{E}[\eta(\y)|\theta]\right)^{\top}|\theta\right]\right)^{1/2}.
$$For $\Sigma(\theta):=A_n(\theta)^{\top}A_n(\theta)$, the BSL ``likelihood'' is then given by $${g}_n(\eta|\theta):=\left(2\pi\right)^{-\frac{d_\eta}{2}}\text{det}\left(\Sigma(\theta)\right)^{-1}\exp\left(-\frac{1}{2}\left\{\eta-b(\theta) \right\}^{\top}\Sigma^{-1}(\theta)\left\{\eta-b(\theta) \right\}\right).
$$Even though $b$ and $\Sigma$ can depend on $n$, we suppress this dependence for notational simplicity. 

The following result is a modification of Corollary 1 in \cite{marin2014relevant}, and the proof follows similarly.  

\begin{lemma}\label{lem:one}
	Under Assumption \ref{ass:one}-\ref{ass:six}, the idealized BSL posterior $\pi\{\cdot|\eta(\y)\}$ concentrates at the
	rate $1/v_n$ onto $b_0$, provided that $\alpha > \tau$ .
\end{lemma}
\begin{proof}
	Take $M_n$ to be a sequence diverging to $\infty$. Define $$m(\eta):=\int_{\Theta}g_n\left(\eta|\theta\right)\pi(\theta)\dt\theta, $$ and consider the set $$T_n(M_n):=\left\{\theta\in\Theta:\|b(\theta)-b_0\|>M_nv_n^{-1}\right\}.$$ Now, consider the BSL posterior over the set $T_n(M_n)$:
	\begin{flalign*}
	\frac{\int_{T_n(M_n)}g_n\left[\eta|\theta\right]\pi(\theta)\dt\theta}{m\left(\eta\right)} \equiv \frac{\int_{T_n(M_n)}\frac{g_n\left[\eta|\theta\right]}{g^0_n\left(\eta\right)}\pi(\theta)\dt\theta}{{m\left(\eta\right)}/{g^0_n\left(\eta\right)}}.
	\end{flalign*}Define $$N_n:=\int_{T_n(M_n)}\frac{g_n\left[\eta|\theta\right]}{g^0_n\left(\eta\right)}\pi(\theta)\dt\theta$$ and $$D_n:=m\left(\eta\right)/g^0_n\left(\eta\right).$$ By Lemma \ref{lem:two}, we have that $D_n\gtrsim v_n^{-\tau}$. Moreover, by Lemma \ref{lem:three}, we have that $N_n\lesssim M_n^{-\alpha}v_n^{-\alpha}$. Therefore, $$\Pi\left\{T_n(M_n)\mid \eta\right\}=\frac{N_n}{D_n}=o_{P}\left(v_{n}^{-\alpha}/v_{n}^{-\tau}\right)=o_{P}(1).$$
	From the above, conclude that  $$\Pi\left[\|b(\theta)-b_0\|>M_nv_{n}^{-1}|\eta\right]=o_P(1).$$ Applying Assumption \ref{ass:five}, we have the stated result: $$\Pi\left[\|\theta-\theta_0\|>L\{M_nv_{n}^{-1}\}^{\kappa}|\eta\right]=o_P(1).$$
\end{proof}
\begin{lemma}\label{lem:two}
	Under Assumptions \ref{ass:one}-\ref{ass:six},  $$\lim_{n\rightarrow\infty}P^0_n\left(m\left(\eta\right)/g^0_n\left(\eta\right)\gtrsim v_n^{-\tau}\right)=1.$$
\end{lemma}
\begin{proof}
	Fix $\delta>0$. By Assumption \ref{ass:one}, there exists an $M_\delta$ such that $$P^0_n\left\{v_n\|\eta-b_0\|>M_\delta\right\}<\delta.$$ For all $\epsilon>0$, by Assumption \ref{ass:six}, there exists $U_\epsilon, \delta_\epsilon$ such that, for $\eta\in V_n$, 
	\begin{flalign*}
	\int_{\Theta}{g}_n(\eta\mid\theta)\pi(\theta)\dt\theta\geqslant \delta_\epsilon g^0_n(\eta)\pi\left[\mathcal{F}_n(U_\epsilon)\right],\text{ for }\mathcal{F}_n(u):=\left\{\theta\in\Theta:\|b(\theta)-b_0\|\leq uv_n^{-1}\right\}.
	\end{flalign*}Apply Assumption \ref{ass:four} to obtain 
	\begin{flalign*}
	\int_{\Theta}{g}_n(\eta\mid\theta)\pi(\theta)\dt\theta\geqslant \delta_\epsilon g^0_n(\eta)v_{n}^{-\tau}.
	\end{flalign*}From the definition of $V_n$, it follows that, for $n$ large enough, there exists some $c(\varepsilon)$ such that  $$P^0_n\left\{\int_{\Theta}{g}_n(\eta\mid\theta)\pi(\theta)\dt\theta\geqslant c(\varepsilon) g^0_n(\eta)v_{n}^{-\tau}\right\}\geqslant1-\varepsilon.$$
	
\end{proof}
\begin{lemma}\label{lem:three}
	Under Assumptions \ref{ass:one}-\ref{ass:six}, for the set  $T_n(M_n):=\left\{\theta\in\Theta:\|b(\theta)-b_0\|>M_nv_n^{-1}\right\},$ $$\lim_{n\rightarrow\infty}P^0_n\left(\int_{T_n(M_n)}\frac{g_n\left[\eta|\theta\right]}{g^0_n\left(\eta\right)} \pi(\theta){{\text{d}}\theta}\lesssim M_n^{-\alpha}v_n^{-\alpha}\right)=1.
	$$
\end{lemma}
\begin{proof}Fix $\delta>0$. By Assumption \ref{ass:one}, there exists some $M_\delta$ such that, for some $n$ large enough, $$P^0_n \left\{ \left\| \eta - b _ { 0 } \right\| > M _ { \delta } / v _ { n } \right\} < \delta.$$ On this set consider the joint probability $$P^0_ { n } \left\{ \int _ { {T } _ { n }(M_\delta) } {g} _ { n} \left( \eta | \theta _ {  } \right) \pi _ {  } \left( \theta _ {  } \right) \mathrm { d } \theta _ {  } > g^0 _ { n } \left( \eta \right) v_n^{-\tau} \right\}. $$
	Applying Markov inequality and Fubini's theorem
	\begin{flalign*}
	P^0 _ { n }& \left\{ \int _ { {T } _ { n }(M_\delta) } {g} _ { n} \left( \eta | \theta _ {  } \right) \pi _ {  } \left( \theta _ {  } \right) \mathrm { d} \theta _ {  } >  g^0_ { n } \left( \eta \right) v_n^{-\tau} \right\}\\
	&\leqslant P^0_n \left( \left\| \eta - b_ { 0 } \right\| > M _ { \delta } v _ { n } ^ { - 1 } \right) + v_n^{\tau}\int _ { {T } _ { n }(M_\delta)} \int _ { \left\| \mathbf { S } _ { n }- b(\theta) \right\| \geqslant M _ { \delta } } \frac { 1 } { g^0 _ { n } ( t )  } g^0 _ { n } ( t ){g}_ { n } \left( t | \theta _ {  } \right) \pi _ { } \left( \theta _ {  } \right) \mathrm { d } t \mathrm { d } \theta _ { }\\ 
	&{ \leqslant P^0_n\left( \left\| \eta-b_0 \right\| > M _ { \delta } v _ { n } ^ { - 1 } \right) + v_{n}^{\tau}\int _ { {T } _ { n }(M_\delta) }  \int _ { \left\| \eta- b(\theta) \right\| \geqslant M _ { \delta }} {g} _ { n } \left( t | \theta _ {  } \right) \pi _ {  } \left( \theta _ {  } \right) \mathrm { d } t \mathrm { d } \theta _ {  } }
	\\ 
	&{ \leqslant P^0_n \left( \left\| \eta - b_0 \right\| > M _ { \delta } v _ { n } ^ { - 1 } \right) + v_{n}^{\tau}\int _ { {T } _ { n }(M_\delta) }  {G} _ { n } \left[ \left\| \eta- b(\theta) \right\| \geqslant M _ { \delta } | \theta  \right] \pi _ {  } \left( \theta _ {  } \right) \mathrm { d } \theta _ {  } }.
	\end{flalign*}
	It then follows that, by Assumption \ref{ass:three}, 
	\begin{flalign*}
	{P^0_n \left( \left\| \eta-b_0 \right\| > M _ { \delta } v _ { n } ^ { - 1 } \right) + v_{n}^{\tau}\int _ { {T } _ { n }(M_\delta) }  {G} _ { n } \left[ \left\| \eta- b(\theta) \right\| \geqslant M _ { \delta } | \theta  \right] \pi _ {  } \left( \theta _ {  } \right) \mathrm { d } \theta _ {  } }&\\\leqslant v_n^{\tau}\left(M_\delta v_{n}\right)^{-\alpha}\int_{\Theta}c(\theta)\pi(\theta)\dt\theta\lesssim o\left(v_{n}^{\tau-\alpha}\right),
	\end{flalign*}where, by Assumption \ref{ass:three}, $\int_{\Theta}c(\theta)\pi(\theta)\dt \theta=O(1)$. Conclude that, for any $\delta>0$, $$P^0_ { n } \left\{ \int _ { {T } _ { n }(M_\delta) } {g} _ { n} \left( \eta | \theta _ {  } \right) \pi _ {  } \left( \theta _ {  } \right) \mathrm { d } \theta _ {  } >  g^0_ { n } \left( \eta \right) v_n^{-\tau}\right\}\leqslant  \delta+
	o(v_n^{\tau-\alpha})\leqslant 2\delta$$ for $n$ large enough. 
\end{proof}

\begin{proof}[Proof of Proposition \ref{prop:one}]
	
	First, we prove the stated result for the posterior of $\Gamma$.

	\noindent\textbf{Part (1): $\Pi[\Gamma\in A|\eta(\y)]=\Pi[\Gamma\in A]+o_P(1)$.}
	We prove the result for R-BSL-M and R-BSL-V separately. 
	\bigskip

	\noindent\textbf{R-BSL-M:} From the posterior concentration of the BSL posterior we have that, for some $M_n\rightarrow\infty$, with $M_n/v_n\rightarrow0$, for $A\subseteq\mathcal{G}$,
	\begin{flalign*}
	&\Pi\left[A|\eta(\y)\right]\\&=\frac{\int_{A}\int_{\|b-b_0\|\leq M_n/v_n}(2\pi)^{-\frac{d_\eta}{2}}\text{det}\left[\Sigma(b)\right]^{-1/2}e^{\left(-\frac{1}{2}\left[b-\eta-x(\Gamma)\right]^{\top}\left[\Sigma(b)\right]^{-1}\left[b-\eta-x(\Gamma)\right]\right)}\pi(b)\pi(\Gamma) \dt b \dt \Gamma}{\int \int_{\|b-b_0\|\leq M_n/v_n}(2\pi)^{-\frac{d_\eta}{2}}\text{det}\left[\Sigma(b)\right]^{-1/2}e^{\left(-\frac{1}{2}\left[b-\eta-x\right]^{\top}\left[\Sigma(b)\right]^{-1}\left[b-\eta-x\right]\right)}\pi(b)\pi(\Gamma)\dt b \dt \Gamma}+o_P(1),
	\end{flalign*}where $x(\Gamma):=\text{diag}[\Sigma(b)]^{1/2}\Gamma$ and $\eta:=\eta(\y)$.
	
	Define $Z:=v_n\left(b-\eta\right)$ and consider the change of variables $b\mapsto v_n\left(b-\eta\right)+v_n\left(\eta-b_0\right)\equiv Z+v_n\left(\eta-b_0\right)$,  which yields
	\begin{flalign*}
	\Pi\left[A|\eta(\y)\right]&=\frac{N_n}{D_n}+o_P(1),
	\end{flalign*}where 
	\begin{flalign*}
	N_n&:={\int_{A}\int_{\|Z\|\leq {M}_n}\frac{e^{\left(-\frac{1}{2}\left[Z/v_n-x(\Gamma)\right]^{\top}\left[\Sigma(Z/v_n+\eta)\right]^{-1}\left[Z/v_n-x(\Gamma)\right]\right)}}{(2\pi)^{d/2}\text{det}\left[\Sigma(Z/v_n+\eta)\right]^{1/2}}\pi(Z/v_n+\eta)\pi(\Gamma) \dt Z \dt \Gamma},
	\end{flalign*}and 
	\begin{flalign*}
	D_n&:={\int \int_{\|Z\|\leq {M}_n}\frac{e^{\left(-\frac{1}{2}\left[Z/v_n-x(\Gamma)\right]^{\top}\left[\Sigma(Z/v_n+\eta)\right]^{-1}\left[Z/v_n-x(\Gamma)\right]\right)}}{(2\pi)^{\frac{d_\eta}{2}}\text{det}\left[\Sigma(Z/v_n+\eta)\right]^{1/2}}\pi(Z/v_n+\eta)\pi(\Gamma)\dt Z \dt \Gamma}.
	\end{flalign*}

	We now analyze $N_n$ and $D_n$ separately. 
	\medskip 
	
	\noindent\textbf{Term $D_n$:} By Assumption \ref{ass:one}, $P^0_n\left\{\|\eta-b_0\|\leq M_{\delta}/ v_n\right\}\geq 1-\delta$ for some $M_\delta,\delta>0$. On this set,  from the definition of $M_n$, 
	\begin{flalign*}
	\sup_{\|Z\|\leq M_n}\left\|\pi(Z/v_n+\eta)-\pi(b_0)\right\|&=o_P(1)\\\sup_{\|Z\|\leq M_n}\left\|\Sigma(Z/v_n+\eta)-\Sigma(b_0)\right\|&=o_P(1).
	\end{flalign*}The first equation follows from continuity of $\pi(\cdot)$ and the second from continuity of $A_n(\theta)\equiv \Sigma^{1/2}(\theta)$. By the dominated convergence theorem, 
	$$
	\frac{D_n}{\pi(b_0)\text{det}\left[\Sigma(b_0)\right]^{-1/2}}= {\int \int_{\|Z\|\leq {M}_n}(2\pi)^{-\frac{d_\eta}{2}}e^{-\frac{1}{2}\left[Z/v_n-x(\Gamma)\right]^{\top}\left[\Sigma(b_0)\right]^{-1}\left[Z/v_n-x(\Gamma)\right]}\pi(\Gamma)\dt Z \dt \Gamma}+o_P(1).
	$$Define $\tilde{Z}:=\left[\Sigma^{1/2}(b_0)v_n\right]^{-1}Z$ and note that, by Assumption \ref{ass:two},  $\tilde{Z}:=A(b_0)^{-1}Z+o_P(1)$, for some positive definite matrix $A(b_0)$. We have that
	$$
	\frac{D_n}{\pi(b_0)\text{det}\left[\Sigma(b_0)\right]^{-1/2}}= {\int (2\pi)^{-\frac{d_\eta}{2}}\int_{\|\tilde{Z}\|\leq {M}_n}e^{-\frac{1}{2}\tilde{Z}^{\top}\tilde{Z}+\tilde{Z}^{\top}\Gamma}e^{-\frac{1}{2}\Gamma^{\top}\Gamma}\pi(\Gamma)\dt Z \dt \Gamma}+o_P(1).
	$$Recall the following:  for $x,y\in\mathbb{R}^{d_\eta}$,
	$$\int_{\mathbb{R}^{d_\eta}}e^{-\frac{1}{2}x^{\top}x+x^{\top}y }\dt x= (2\pi)^{d_\eta/2}e^{y^{\top}y/2}.
	$$ 
	From the above fact and the  dominated convergence theorem
	$$\int_{\|\tilde{Z}\|\leq {M}_n}e^{ -\frac{1}{2}\tilde{Z}^{\top}\tilde{Z}+\tilde{Z}^{\top}\Gamma }\dt Z\rightarrow(2\pi)^{d_\eta/2}e^{\frac{1}{2}\Gamma^{\top}\Gamma}.$$ Apply the above and Fubini's Theorem to deduce
	$$
	\frac{D_n}{\pi(b_0)\text{det}\left[\Sigma(b_0)\right]^{-1/2}}\rightarrow_p\int \pi(\Gamma)\dt\Gamma=1.
	$$
	\medskip 
	
	\noindent\textbf{Term $N_n$:} Apply the same argument as for $D_n$, to obtain
	$$
	\frac{N_n}{\pi(b_0)\text{det}\left[\Sigma(b_0)\right]^{-1/2}}= {\int_{A} (2\pi)^{-\frac{d_\eta}{2}}\int_{\|\tilde{Z}\|\leq {M}_n}e^{-\frac{1}{2}\tilde{Z}^{\top}\tilde{Z}+\tilde{Z}^{\top}\Gamma }e^{-\frac{1}{2}\Gamma^{\top}\Gamma}\pi(\Gamma)\dt Z \dt \Gamma}\rightarrow_{p}\int_{A}\pi(\Gamma)\dt \Gamma.
	$$Conclude that 
	$$
	\Pi\left[A|\eta(\y)\right]=\frac{N_n}{D_n}+o_P(1)=\Pi[A]+o_{P}(1).
	$$

	\noindent\textbf{R-BSL-V:} Define $$V(b,\Gamma):=\Sigma(b)+\text{diag}\left[\Sigma^{1/2}(b)\right]\Gamma\Gamma^{\top}\text{diag}\left[\Sigma^{1/2}(b)\right]^{\top}$$
	From posterior concentration, for some $M_n\rightarrow\infty$, with $M_n/v_n\rightarrow0$, 
	\begin{flalign*}
	&\Pi\left[A|\eta\right]\\&=\frac{\int_{A}\int_{\|b-b_0\|\leq M_n/v_n}(2\pi)^{-\frac{d_\eta}{2}}\text{det}\left[V(b,\Gamma)\right]^{-1}e^{\left(-\frac{1}{2}\left[b-\eta\right]^{\top}V(b,\Gamma)^{-1}\left[b-\eta\right]\right)}\pi(b)\pi(\Gamma) \dt b \dt \Gamma}{\int \int_{\|b-b_0\|\leq M_n/v_n}(2\pi)^{-\frac{d_\eta}{2}}\text{det}\left[V(b,{\Gamma})\right]^{-1}e^{\left(-\frac{1}{2}\left[b-\eta\right]^{\top}V(b,{\Gamma})^{-1}\left[b-\eta\right]\right)}\pi(b)\pi(\Gamma)\dt b \dt \Gamma}+o_P(1),
	\end{flalign*}where $\eta:=\eta(\y)$. 
	
	Similar to the case of R-BSL-M, define $Z:=v_n\left(b-\eta\right)$ and consider the change of variables $b\mapsto v_n\left(b-\eta\right)+v_n\left(\eta-b_0\right)\equiv Z+v_n\left(\eta-b_0\right)$,  which yields, for some ${M}_n\rightarrow\infty$,
	\begin{flalign*}
	&\Pi\left[A|\eta(\y)\right]=\frac{N_n}{D_n}+o_P(1)
	\end{flalign*}
	where
	\begin{flalign*}N_n&:={\int_{A}\int_{\|Z\|\leq {M}_n}\frac{e^{\left(-\frac{1}{2}\left[Z/v_n\right]^{\top}V^{-1}(Z/v_n+\eta,\Gamma)^{}\left[Z/v_n\right]\right)}}{(2\pi)^{\frac{d_\eta}{2}}\text{det}\left[V(Z/v_n+\eta,\Gamma)\right]^{1/2}}\pi(Z/v_n+\eta)\pi(\Gamma) \dt Z \dt \Gamma}\end{flalign*}and
	\begin{flalign*}D_n&:={\int \int_{\|Z\|\leq {M}_n}\frac{e^{\left(-\frac{1}{2}\left[Z/v_n\right]^{\top}V^{-1}(Z/v_n+\eta,\Gamma)^{}\left[Z/v_n\right]\right)}}{(2\pi)^{\frac{d_\eta}{2}}\text{det}\left[V(Z/v_n+\eta,\Gamma)\right]^{1/2}}\pi(Z/v_n+\eta)\pi(\Gamma)\dt Z \dt \Gamma}.
	\end{flalign*}
	Again, we analyze $N_n$ and $D_n$ separately. 
	\medskip 
	
	\noindent\textbf{Term $D_n$:} Similar to the previous result, by Assumption \ref{ass:one}, $P _ { n } ^ { 0 } \left\{ \left\| \eta - b _ { 0 } \right\| \leq M _ { \delta } / v _ { n } \right\} \geq 1 - \delta$ for some $M_\delta,\delta>0$. Similarly to the case of R-BSL-M,
	\begin{flalign*}
	\sup_{\|Z\|\leq M_n}\left\|\pi(Z/v_n+\eta)-\pi(b_0)\right\|&=o_P(1)\\\sup_{\|Z\|\leq M_n}\left\|\Sigma(Z/v_n+\eta)-\Sigma(b_0)\right\|&=o_P(1).
	\end{flalign*}By Assumption \ref{ass:two}, and the dominated convergence theorem,
	$$
	\frac{D_n}{\pi(b_0)v_n^{-d_\eta}}= {\int \int_{\|Z\|\leq {M}_n}(2\pi)^{-\frac{d_\eta}{2}}{\text{det}\left[v^2_nV(b_0,\Gamma)\right]}^{-1/2}e^{\left(-\frac{1}{2}\left[Z/v_n\right]^{\top}\left[V(b_0,\Gamma)\right]^{-1}\left[Z/v_n\right]\right)}\pi(\Gamma)\dt Z \dt \Gamma}+o_P(1).
	$$Define $$\tilde{Z}:=\left[v^2_n\Sigma(b_0)+\text{diag}\left[v_n\Sigma^{1/2}(b_0)\right]\Gamma\Gamma^{\top}\text{diag}\left[v_n\Sigma^{1/2}(b_0)\right]^{\top}\right]^{-1/2}Z.$$ By Assumption \ref{ass:two}, $\tilde{Z}= A(b_0+\Gamma)^{-1/2}Z+o_P(1)$ and we then obtain
	$$
	\frac{D_n}{\pi(b_0)v_n^{-d_\eta}}= {\int (2\pi)^{-\frac{d_\eta}{2}}\int_{\|\tilde{Z}\|\leq {M}_n}\text{det}\left[A(b_0+\Gamma)\right]^{-1/2}e^{\left(-\frac{1}{2}\tilde{Z}^{\top}\tilde{Z}\right)}\pi(\Gamma)\dt Z \dt \Gamma}+o_P(1).
	$$From the dominated convergence theorem,  $$(2\pi)^{-d_\eta/2}\int_{\|\tilde{Z}\|\leq {M}_n}\text{det}\left[A(b_0+\Gamma)\right]^{-1/2}e^{\left(-\frac{1}{2}\tilde{Z}^{\top}\tilde{Z}\right)}\dt Z\rightarrow_p1.$$ Applying the above and Fubini's Theorem delivers
	$$
	\frac{D_n}{\pi(b_0)v_n^{-d_\eta}}\rightarrow_p\int \pi(\Gamma)\dt\Gamma=1.
	$$
	\medskip 
	
	\noindent\textbf{Term $N_n$:} Apply the same argument as for $D_n$, to deduce that 
	$$
	\frac{N_n}{\pi(b_0)v_n^{-d_\eta}}= {\int_{A} (2\pi)^{-\frac{d_\eta}{2}}\int_{\|\tilde{Z}\|\leq {M}_n}\text{det}\left[A(b_0+\Gamma)\right]^{-1/2}e^{\left(-\frac{1}{2}\tilde{Z}^{\top}\tilde{Z}\right)}\pi(\Gamma)\dt Z \dt \Gamma}+o_P(1)\rightarrow_{p}\int_{A}\pi(\Gamma)\dt \Gamma.
	$$Conclude that 
	$$
	\Pi\left[A|\eta(\y)\right]=\frac{N_n}{D_n}+o_P(1)=\Pi[A]+o_{P}(1).
	$$ 
	
	From the arbitrary nature of $A$, the result holds for any $A\subseteq\mathcal{G}$ by the absolute continuity of $\Pi[\cdot]$. 
	
	\noindent\textbf{Part (2): $\Pi[\|b(\theta)-b_0\|\leq\delta|\eta(\y)]=1+o_P(1)$.}
	The result follows along similar lines to those given for \textbf{Part (1)}, the posterior for the $\Gamma$ components, and, hence, we only sketch the result and only for the R-BSL-M case. The equivalent result for R-BSL-V is very similar and hence omitted. 
	
	From the posterior concentration of the BSL posterior in Lemma \ref{lem:one} we have that, for any $\delta>0$, and some $M_n\rightarrow\infty$, as $n\rightarrow\infty$, such that $\lim_nM_n/v_n=\delta$,
	\begin{flalign*}
	&\Pi\left[\|b(\theta)-b_0\|\leq\delta|\eta(\y)\right]\\&=\frac{\int_{\|b(\theta)-b_0\|\leq\delta}\int_{}(2\pi)^{-\frac{d_\eta}{2}}\text{det}\left[\Sigma(b)\right]^{-1/2}e^{\left(-\frac{1}{2}\left[b-\eta-x(\Gamma)\right]^{\top}\left[\Sigma(b)\right]^{-1}\left[b-\eta-x(\Gamma)\right]\right)}\pi(b)\pi(\Gamma) \dt b \dt \Gamma}{\int \int_{}(2\pi)^{-\frac{d_\eta}{2}}\text{det}\left[\Sigma(b)\right]^{-1/2}e^{\left(-\frac{1}{2}\left[b-\eta-x\right]^{\top}\left[\Sigma(b)\right]^{-1}\left[b-\eta-x\right]\right)}\pi(b)\pi(\Gamma)\dt b \dt \Gamma},\\&=\frac{\int_{v_n\|b(\theta)-b_0\|\leq M_n}\int_{}(2\pi)^{-\frac{d_\eta}{2}}\text{det}\left[\Sigma(b)\right]^{-1/2}e^{\left(-\frac{1}{2}\left[b-\eta-x(\Gamma)\right]^{\top}\left[\Sigma(b)\right]^{-1}\left[b-\eta-x(\Gamma)\right]\right)}\pi(b)\pi(\Gamma) \dt b \dt \Gamma}{\int \int_{v_n\|b-b_0\|\leq M_n}(2\pi)^{-\frac{d_\eta}{2}}\text{det}\left[\Sigma(b)\right]^{-1/2}e^{\left(-\frac{1}{2}\left[b-\eta-x\right]^{\top}\left[\Sigma(b)\right]^{-1}\left[b-\eta-x\right]\right)}\pi(b)\pi(\Gamma)\dt b \dt \Gamma}+o_P(1),
	\end{flalign*}where $x(\Gamma):=\text{diag}[\Sigma(b)]^{1/2}\Gamma$ and $\eta:=\eta(\y)$.
	
	Define $Z:=v_n\left(b-\eta\right)$ and consider the change of variables $b\mapsto v_n\left(b-\eta\right)+v_n\left(\eta-b_0\right)\equiv Z+v_n\left(\eta-b_0\right)$,  which yields
	\begin{flalign*}
	\Pi\left[\|b(\theta)-b_0\|\leq\delta|\eta(\y)\right]&=\frac{N_n}{D_n}+o_P(1),
	\end{flalign*}where 
	\begin{flalign*}
	N_n&:={\int_{\|Z\|\leq {M}_n}\int\frac{e^{\left(-\frac{1}{2}\left[Z/v_n-x(\Gamma)\right]^{\top}\left[\Sigma(Z/v_n+\eta)\right]^{-1}\left[Z/v_n-x(\Gamma)\right]\right)}}{(2\pi)^{d/2}\text{det}\left[\Sigma(Z/v_n+\eta)\right]^{1/2}}\pi(Z/v_n+\eta)\pi(\Gamma) \dt \Gamma\dt Z},
	\end{flalign*}and 
	\begin{flalign*}
	D_n&:={\int_{\|Z\|\leq {M}_n}\int\frac{e^{\left(-\frac{1}{2}\left[Z/v_n-x(\Gamma)\right]^{\top}\left[\Sigma(Z/v_n+\eta)\right]^{-1}\left[Z/v_n-x(\Gamma)\right]\right)}}{(2\pi)^{\frac{d_\eta}{2}}\text{det}\left[\Sigma(Z/v_n+\eta)\right]^{1/2}}\pi(Z/v_n+\eta)\pi(\Gamma) \dt \Gamma\dt Z}.
	\end{flalign*}

	We now analyze $N_n$ and $D_n$ separately, starting with $D_n$. 
	
	\noindent\textbf{Term $D_n$.} The argument is identical to that used to prove the same part in \textbf{Part (1)} of the result. Therefore, defining $\tilde{Z}:=\left[\Sigma^{1/2}(b_0)v_n\right]^{-1}Z$, by Assumption \ref{ass:two}  $\tilde{Z}:=A(b_0)^{-1}Z+o_P(1)$, for some positive definite matrix $A(b_0)$, and it follows that
	\begin{flalign*}
	\frac{D_n}{\pi(b_0)\text{det}\left[\Sigma(b_0)\right]^{-1/2}}&= \int (2\pi)^{-\frac{d_\eta}{2}}\int_{\|\tilde{Z}\|\leq {M}_n}e^{-\frac{1}{2}\tilde{Z}^{\top}\tilde{Z}+\tilde{Z}^{\top}\Gamma}e^{-\frac{1}{2}\Gamma^{\top}\Gamma}\pi(\Gamma)\dt \tilde{Z} \dt \Gamma+o_P(1)\\&\rightarrow_p\int \pi(\Gamma)\dt\Gamma=1.
	\end{flalign*}
	Similar to the proof of  \textbf{Part (1)}, the second line follows from the dominated convergence theorem and Fubini's theorem.

	\noindent\textbf{Term $N_n$.} Again, the argument follows similarly to that used for the posterior of the $\Gamma$ components. Similar to the case for $D_n$, we can obtain
	\begin{flalign*}
	\frac{N_n}{\pi(b_0)\text{det}\left[\Sigma(b_0)\right]^{-1/2}}&= (2\pi)^{-\frac{d_\eta}{2}}\int_{\|\tilde{Z}\|\leq {M}_n}{\int_{} e^{-\frac{1}{2}\tilde{Z}^{\top}\tilde{Z}+\tilde{Z}^{\top}\Gamma }e^{-\frac{1}{2}\Gamma^{\top}\Gamma}\pi(\Gamma)\dt \tilde{Z} \dt \Gamma}\\&= (2\pi)^{-\frac{d_\eta}{2}}\int_{\|\tilde{Z}\|\leq {M}_n}{\int_{} e^{-\frac{1}{2}\tilde{Z}^{\top}\tilde{Z}+\tilde{Z}^{\top}\Gamma }e^{-\frac{1}{2}\Gamma^{\top}\Gamma}\pi(\Gamma)\dt \tilde{Z} \dt \Gamma}
	\end{flalign*}From the dominated convergence theorem, 
	\begin{flalign*}
	\frac{N_n}{\pi(b_0)\text{det}\left[\Sigma(b_0)\right]^{-1/2}}&\xrightarrow{P} (2\pi)^{-\frac{d_\eta}{2}}\int_{}\int_{} e^{-\frac{1}{2}\tilde{Z}^{\top}\tilde{Z}+\tilde{Z}^{\top}\Gamma }e^{-\frac{1}{2}\Gamma^{\top}\Gamma}\pi(\Gamma)\dt\tilde{Z} { \dt \Gamma},
	\end{flalign*}for $\tilde{Z}$ a standard Gaussian vector. Again, recalling that, for $x,y\in\mathbb{R}^{d_\eta}$,
	$$\int_{}e^{-\frac{1}{2}x^{\top}x+x^{\top}y }\dt x= (2\pi)^{d_\eta/2}e^{y^{\top}y/2},
	$$ we can apply the above and Fubini's theorem to obtain that 
	\begin{flalign*}
	\frac{N_n}{\pi(b_0)\text{det}\left[\Sigma(b_0)\right]^{-1/2}}\xrightarrow{P}(2\pi)^{-\frac{d_\eta}{2}}\int\int_{}e^{ -\frac{1}{2}\tilde{Z}^{\top}\tilde{Z}+\tilde{Z}^{\top}\Gamma }e^{-\frac{1}{2}\Gamma^{\top}\Gamma}\pi(\Gamma)\dt \tilde{Z} \dt \Gamma=\int\pi(\Gamma)\dt\Gamma=1.
	\end{flalign*}
	
	Putting the two terms together we have that 
	$$
	\Pi\left[\|b(\theta)-b_0\|\leq\delta|\eta(\y)\right]=\frac{N_n}{D_n}+o_P(1),
	$$as stated.

\end{proof}

\bibliographystyle{apalike} 
\bibliography{refs}

\end{document}